\documentclass[10pt, journal, final]{IEEEtran}

\usepackage{cite}

\usepackage[hidelinks]{hyperref}

\ifCLASSINFOpdf
   \usepackage[pdftex]{graphicx}

   \graphicspath{{/}}

\else

   \usepackage[dvips]{graphicx}

   \graphicspath{{/}}

\fi

\usepackage{amsmath}

\usepackage{amsfonts}
\usepackage{amsthm}
\usepackage{amssymb}
\usepackage{stmaryrd}
\usepackage{braket}

\usepackage{enumitem}
\usepackage{dsfont}

\makeatletter
\newtheorem*{rep@theorem}{\rep@title}
\newcommand{\newreptheorem}[2]{%
	\newenvironment{rep#1}[1]{%
		\def\rep@title{#2 \ref{##1}}%
		\begin{rep@theorem}}%
		{\end{rep@theorem}}}
\makeatother

\newtheorem{prop}{Proposition}
\newreptheorem{prop}{Proposition}
\newtheorem{lem}{Lemma}
\newreptheorem{lem}{Lemma}
\newtheorem{defi}{Definition}

\newcommand{\bs}[1]{\boldsymbol{#1}}
\newcommand{\e}{\mathrm{e}}
\newcommand{\id}{\mathds{1}}

\newcounter{MYtempeqncnt}

\usepackage{array}

\ifCLASSOPTIONcompsoc
  \usepackage[caption=false,font=normalsize,labelfont=sf,textfont=sf]{subfig}
\else
  \usepackage[caption=false,font=footnotesize]{subfig}
\fi

\usepackage{url}

\usepackage{xcolor}

\begin{document}

\title{Quantum Pin Codes}

\author{Christophe~Vuillot
        and~Nikolas~P.~Breuckmann

\thanks{C. Vuillot, previously QuTech, TU Delft, now Inria, France, \href{mailto:christophe.vuillot@inria.fr}{christophe.vuillot@inria.fr}.}%
\thanks{N. P. Breuckmann, University College London, \href{mailto:n.breuckmann@ucl.ac.uk}{n.breuckmann@ucl.ac.uk}.}
\thanks{ This paper was presented at QEC19 in London.}%
}

\maketitle

\begin{abstract}
	We introduce quantum pin codes: a class of quantum CSS codes.
	Quantum pin codes are a generalization of quantum color codes and Reed-Muller codes and share a lot of their structure and properties.
	Pin codes have gauge operators, an unfolding procedure and their stabilizers form so-called $\ell$-orthogonal spaces meaning that the joint overlap between any $\ell$ stabilizer elements is always even.
	This last feature makes them interesting for devising magic-state distillation protocols, for instance by using puncturing techniques.
	We study examples of these codes and their properties.
\end{abstract}

\section{Introduction}
\IEEEPARstart{T}{he realization} of a fault-tolerant universal quantum computer is a tremendous challenge.
At each level of the architecture, from the hardware implementation up to the quantum software, there are difficult problems that need to be overcome. 
Hovering in the middle of the stack, quantum error correcting codes influence both hardware design and software compilation.
They play a major role not only in mitigating noise and faulty operations but also in devising protocols to distill the necessary resources that grant universality to an error corrected quantum computer \cite{campbell_roads_2017}.
The study and design of quantum error correcting codes is therefore one of the major tasks to be undertaken on the way to universal quantum computation.

A well-studied class of quantum error correcting codes are Calderbank-Shor-Steane codes (CSS codes)  \cite{calderbank_good_1996,steane_andrew_multiple-particle_1996}, which are stabilizer quantum codes \cite{gottesman_stabilizer_1997, calderbank_quantum_1997}.
The advantage of CSS codes over general stabilizer codes is their close connection to linear codes which have been studied in classical coding theory.
A CSS code can be constructed by combining two binary linear codes.
Roughly speaking, one code performs parity checks in the Pauli $X$-basis and the other performs parity checks in the Pauli $Z$-basis.
Not any two binary linear codes can be used: it is necessary that any two pairs of code words from each code space have to have even overlap.
Several families of CSS codes have been devised based on geometrical, homological or algebraic constructions \cite{grassl_codes_1997, freedman_projective_2001, freedman_$z_2$-systolic_2002, kitaev_fault-tolerant_2003, bacon_operator_2006, bombin_topological_2006, kovalev_quantum_2013, couvreur_construction_2013, guth_quantum_2014, tillich_quantum_2014, audoux_tensor_2019, leverrier_quantum_2015}, however, it is still open which parameters can be achieved.

Besides being able to protect quantum information, quantum error correcting codes must also allow for some mechanism to process the encoded information without lifting the protection.
It is always possible to find some operations realizing a desired action on the encoded information but these operations may spread errors in the system.
Therefore we shall only consider operations to be fault-tolerant if they do not spread errors.
For instance, if the operation acts separately on each qubit of a code it cannot spread single qubit errors to multi-qubit errors.
This is called a transversal gate, but not any code admits such gates.
More generally, for many codes in the CSS code family it is possible to fault-tolerantly implement Clifford operations, which are all unitary operations preserving Pauli operators under conjugation.
Clifford operations by themselves do not form a universal gate set.
Several techniques to obtain a universal gate set, by supplementing the non-Clifford $T$ gate to Cliffords for example, have been devised \cite{bravyi_universal_2005, bombin_gauge_2015}, among which magic state distillation is currently the most promising candidate.

In this work we introduce a new class of CSS codes, which we call \emph{quantum pin codes}.
They are inspired from $D$-dimensional quantum color codes \cite{bombin_exact_2007} which are known for their transversal gates.
Quantum pin codes form a large family while at the same time have structured stabilizer generators.
Namely, they form $\ell$-orthogonal spaces, meaning that the common overlap between any $\ell$ stabilizer elements is even.
This structure is necessary for codes to admit transversal phase gates and it can be leveraged to obtain codes that can be used within magic state distillation protocols.
This structure comes from an underlying multi-ary relation with a single simple property and we show different ways of constructing such relations.
While the whole family is too large to feature interesting transversal gates for all its members it could lead to interesting subfamilies.
Besides one can obtain magic state distillation protocols from any relation defining a quantum pin code by using puncturing techniques. Moreover, using a slightly different code definition, one can obtain a quantum code with transversal $T$ gates realizing a logical circuit of $CCZ$ gates.
Moreover the construction of pin codes differs substantially from previous approaches making it an interesting space to explore further.

In Section~\ref{sec:pincodes}, after introducing some notations and terminology, we define quantum pin codes, explain their relation to quantum color codes and give some concrete approaches to construct them.
In Section~\ref{sec:multiorthogonality}, we discuss the conditions for transversal implementation of phase gates on a CSS code and magic state distillation.
In Section~\ref{sec:prop}, we investigate the properties of pin codes.
Finally, in Section~\ref{sec:examples}, we study concrete examples of pin codes obtained from Coxeter groups and chain complexes.
We also discuss applications for magic state distillation.

\section{Pin codes}
\label{sec:pincodes}

\subsection{Quantum stabilizer codes, CSS codes and subsystem codes}
A quantum code can be defined by considering an Abelian subgroup $S$ of the Pauli group on $n$ qubits such that $-I \not\in S$.
That is to say a commuting group generated by tensor product of the single-qubit Pauli operators and identity:
\begin{eqnarray}
	X = \begin{pmatrix}
		0 & 1 \\ 1 & 0
	\end{pmatrix},& Y = \begin{pmatrix}
	0 & -i \\ i & 0
\end{pmatrix},\\
	Z = \begin{pmatrix}
1 & 0 \\ 0 & -1
\end{pmatrix},& I = \begin{pmatrix}
1 & 0 \\ 0 & 1
\end{pmatrix}.\label{eq:Paulis}
\end{eqnarray}
The quantum code associated with $S$ is defined as the common $+1$-eigenspace of all elements of $S$.
\begin{equation}
	\mathcal{C} = \left \{\ket{\psi}\;\middle\vert\; \forall s\in S,s\ket{\psi} = \ket{\psi}\right \}.
\end{equation}
Quantum codes defined in this way are called \emph{stabilizer codes}, the generator elements of $S$ are called stabilizer generators and generic elements of $S$ simply stabilizers.
A generating set of $S$ defines stabilizer checks which can be measured in order to infer information about the errors.
The logical operators of a stabilizer code are the Pauli operators which commute with all stabilizers.
That is to say they belong to the centralizer of the stabilizer group.
The fact that they commute with stabilizers makes them preserve the code space.
A logical operator which also belongs to the stabilizer group is called trivial as it not only preserve the code space but also each code state.
Logical operators outside of the stabilizer group are non-trivial as they act non-trivially on the code states while preserving the codes space.
For more background on stabilizer codes see \cite{gottesman_stabilizer_1997}.
A stabilizer code is called a Calderbank-Shor Steane (CSS) code \cite{calderbank_good_1996, steane_andrew_multiple-particle_1996} if there exists a set of stabilizer checks which only act non-trivially as either Pauli-$X$ or Pauli-$Z$ on each qubit in their support.
In that case the code is most conveniently described by two classical binary linear codes, $\mathcal{C}_X\subset\mathbb{F}_2^n$ and $\mathcal{C}_Z\subset\mathbb{F}_2^n$, generating the $X$-type and $Z$-type stabilizers respectively.
To form a valid stabilizer code the $X$-type stabilizers and $Z$-type stabilizer must commute which constraints $\mathcal{C}_X$ and $\mathcal{C}_Z$ to be included in each others dual space:
$\mathcal{C}_X\subset\mathcal{C}_Z^\perp$, note that this condition is symmetric.
The logical operators of a CSS code also split into $X$- and $Z$-type operators, they are given by~$\mathcal{C}_Z^\perp$ and $\mathcal{C}_X^\perp$ respectively.
Code states of CSS codes are simply expressed as equal superposition of computational basis states over cosets:
\begin{equation}
	\ket{\overline{x}} = \frac{1}{\sqrt{2^{\vert C_X\vert}}}\sum_{y\in\mathcal{C}_X}\ket{y\oplus x},\label{eq:csscodestate}
\end{equation}
where $x\in\mathcal{C}_Z^\perp$ and $\overline{x}$ designate the coset of $x$ in $\mathcal{C}_Z^\perp/\mathcal{C}_X$.

Subsystem codes are a generalization of stabilizer codes.
A subsystem code is defined by any (non-Abelian) subgroup of the Pauli group~$G$ which is called the \emph{gauge group}.
In order to define the code we define the stabilizer subgroup $S$ of $G$ to be the center of~$G$, i.e. the largest Abelian subgroup of~$G$.
The code is then defined as the stabilizer code of $S$.
The reason to define a code in this way is that a generating set of $G$ (called \emph{gauge checks}) can be of lower weight than any generating elements of $S$.
This allows to infer the eigenvalues of the stabilizer checks by measuring the lower-weight gauge checks.

\subsection{Terminology and formalism}

Consider $D+1$ finite, disjoint sets, $(L_0,\ldots,L_D)$ which we call \emph{levels}. 
The elements in each of the levels are called \emph{pins}. 
If a pin $p$ is contained in a set $L_j$ then $j$ is called the \emph{rank} of~$p$.
Since all the $L_j$ are disjoint each pin has a unique rank.

Consider a $(D+1)$-ary relation on the $D+1$ levels $L_0,\ldots,L_D$, that is to say a subset of their Cartesian product $F\subset L_0\times\cdots\times L_D$.
The tuples in the relation~$F$ will be called \emph{flags}.

A subset of the ranks, $T\subset\{0,\ldots,D\}$, is called a \emph{type}.
We will consider tuples of pins coming from a subset of the levels selected by a type $T$ and call them \emph{collection} of pins of type $T$.
A collection of pins of type $T=\{j_1,\dots,j_k\}$, is therefore an element $s\in L_{j_1}\times\cdots\times L_{j_k}$.
Note that we can interchangeably view a collection of pins as a tuple or a set as long as no two pins come from the same level in the set.

We now define specific subsets of flags, called pinned sets, using projections.
\begin{defi}[Projection of type $T$]
	\label{def:projtt}
	Given a set of flags $F$ and a type $T=\{j_1,\dots,j_k\}$, the projection, $\Pi_T$, is defined as the natural Cartesian product projection acting on the flags
	\begin{eqnarray*}
		\Pi_T :&F&\;\rightarrow \;L_{j_1}\times\cdots\times L_{j_k}\\
		&(p_0,\dots,p_D)&\;\mapsto\;(p_{j_1},\dots,p_{j_k}).
	\end{eqnarray*}
\end{defi}
Note that the projection of empty type, $\Pi_{\emptyset}$, is also well defined: for any $f\in F$ we have $\Pi_{\emptyset}(f) = ()$.
\begin{defi}[Pinned set]
	\label{def:pinset}
	Let $F$ be a set of flags, $s$ be a collection of pins of type~$T$ and $\Pi_T$ be the corresponding projection as defined above.
	We define the pinned set of type~$T$ and collection of pins~$s$, $P_T(s)$, as the preimage of~$s$ under the projection~$\Pi_T$,
	\[P_T(s) = \Pi_T^{-1}\left (s\right )\subset F.\]
\end{defi}
In words: a pinned set is the set of flags whose projection of a given type $T$ yields a given collection of pins, $s$.
A definition of a pinned set which is equivalent to the one given above is 
\begin{align}\label{eq:altcharac}
P_T&(p_{j_{1}},\ldots,p_{j_{k}}) = \nonumber\\
&F\cap\left ( L_0\cdots\times\Set{p_{j_1}}\times\cdots\times\Set{p_{j_k}}\times\cdots L_D\right ).
\end{align}
Proving that these definitions are equivalent amounts to proving
\begin{equation}
\Pi_T^{-1}\left (s\right ) = F\cap\left ( L_0\cdots\times\Set{p_{j_1}}\times\cdots\times\Set{p_{j_k}}\times\cdots L_D\right ).
\end{equation}
To prove the inclusion of the left-hand side in the right-hand side we observe that an element in the set from the left hand-side is in $F$ by definition of $\Pi_T$ and in $\left (L_0\cdots\times\Set{p_{j_1}}\times\cdots\times\Set{p_{j_k}}\times\cdots L_D\right )$ since $s=(p_{j_1},\ldots,p_{j_k})$.
To prove the inclusion of the right-hand side in the left-hand side one takes any element $f\in F\cap\left ( L_0\cdots\times\Set{p_{j_1}}\times\cdots\times\Set{p_{j_k}}\times\cdots L_D\right )$ and notice that $\Pi_T(f) = s$.

The pinned set with respect to the empty type is none other than the full set of flags, $F$.
For convenience, we will refer to a pinned set defined by a collection with $k$ pins as a \mbox{\emph{$k$-pinned set}}.

If one wants to form a mental image one can imagine a pin-board with pins of different colors for each levels on it.
Then the flags can be represented by cords each attached to one pin of each level, see Fig.~\ref{fig:pinsdrawing} as an example.
\begin{figure}[h]
	\centering
	\includegraphics[width=.33\linewidth]{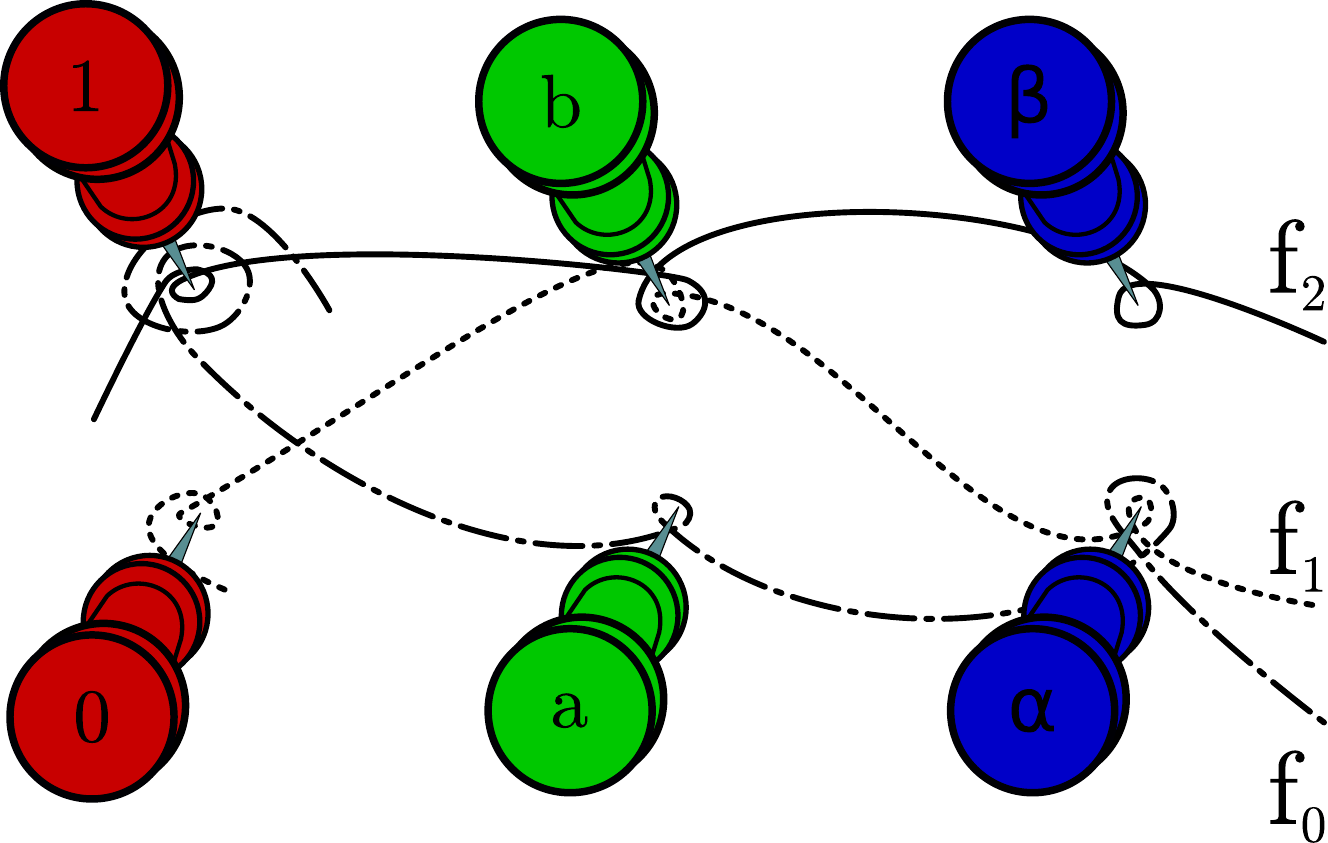}\hfill
	\includegraphics[width=.33\linewidth]{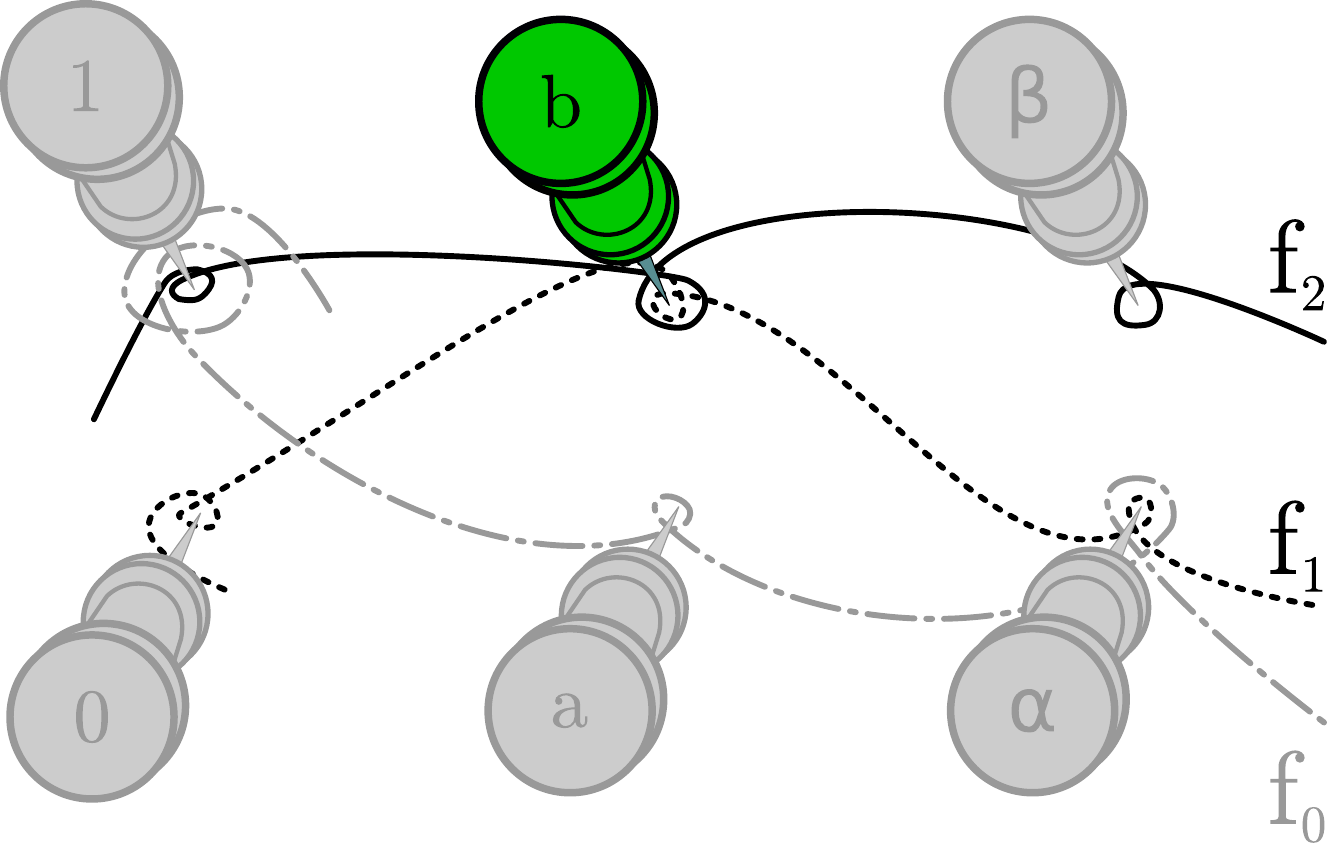}\hfill
	\includegraphics[width=.33\linewidth]{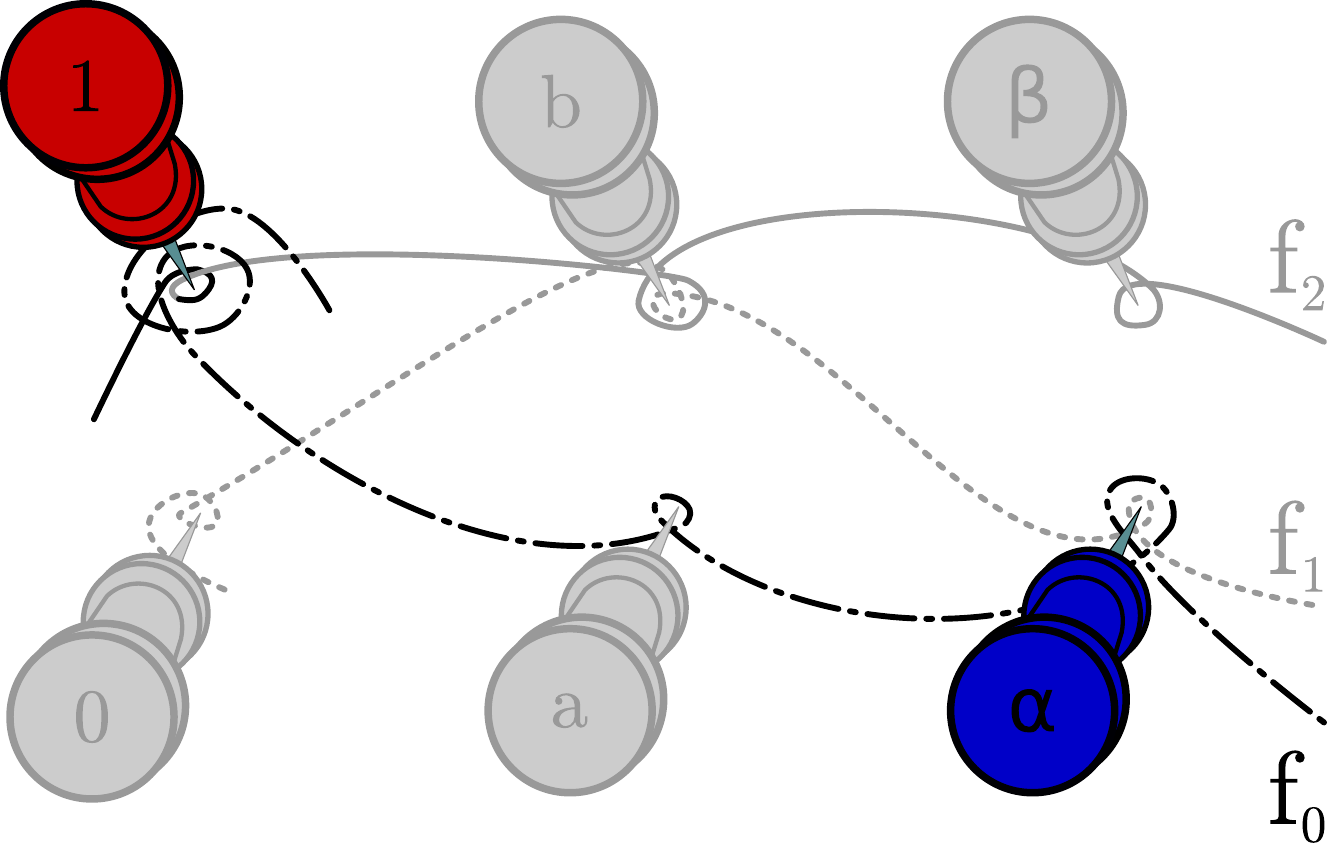}
	
	\caption{Illustration of three levels (red, green and blue) each containing two pins and a relation containing three flags ($f_0$, $f_1$ and $f_2$) symbolized by cords attached to the pins.
	The pinned set $P_{\rm green}(b)$ is composed of the flags $f_1$ and $f_2$.
	The pinned set $P_{({\rm red, blue})}(1,\alpha)$ only contains the flag $f_0$.}
	\label{fig:pinsdrawing}
\end{figure}

The structure of pinned sets layed out above is such that they intersect and decompose nicely.
This is captured by the following two propositions.
\begin{prop}[Intersection of pinned sets]
	\label{prop:interpinset}
	Given $F$, a set of flags, let $s_1$ and $s_2$ be two collections of pins of types $T_1$ and $T_2$ respectively.
	Then the intersection of the two pinned sets $P_{T_1}(s_1)$ and $P_{T_2}(s_2)$ is either empty or a pinned set of type $T_1\cup T_2$ characterized by the collection of pins $s_1\cup s_2$, 
	\begin{align*}
	P_{T_1}(s_1)&\cap P_{T_2}(s_2) = \\&\begin{cases}P_{T_1\cup T_2}(s_1\cup s_2)&\text{ if }\left \vert s_1\cup s_2\right \vert = \left \vert T_1\cup T_2\right \vert\\\emptyset&\text{ otherwise.}\end{cases}
	\end{align*}
\end{prop}
\begin{IEEEproof}
	We write $s_1 = (p_{i_1},\ldots,p_{i_{k_1}})$ and $s_2 = (q_{j_1},\ldots,q_{j_{k_2}})$.
	Consider the case where $T_1\cap T_2 = \emptyset$.
	Then it is always the case that $\left \vert s_1\cup s_2\right \vert = \left \vert T_1\cup T_2\right \vert$ since the levels are disjoint.
	Using the alternative characterization of pinned sets given in \eqref{eq:altcharac} the conclusion then straightforwardly follows.
	Consider the case where $T_1\cap T_2 \neq \emptyset$.
	Either there exist a rank $j\in T_1\cap T_2$ such that $p_j \neq q_j$ then given $f_1\in P_{T_1}(s_1)$ and $f_2\in P_{T_2}(s_2)$ one has $\Pi_j(f_1) = p_j \neq q_j = \Pi_j(f_2)$ hence $f_1\neq f_2$ and so $P_{T_1}(s_1)\cap P_{T_2}(s_2) = \emptyset$.
	In this case it also holds that $\vert s_1\cup s_2\vert > \vert T_1 \cup T_2\vert$ since for every shared rank in $j\in T_1\cap T_2$ for which $p_j\neq q_j$ the difference $\vert s_1\cup s_2\vert - \vert T_1 \cup T_2\vert$ is increased by one.
	In the other case for all $j\in T_1\cap T_2$ we have $p_j = q_j$.
	Hence it holds that $\vert s_1\cup s_2\vert = \vert T_1 \cup T_2\vert$ and $P_{T_1}(s_1)\cap P_{T_2}(s_2) = P_{T_1\cup T_2}(s_1\cup s_2)$ using \eqref{eq:altcharac} again.
\end{IEEEproof}
\begin{prop}[Pinned set decomposition]
	\label{prop:pinsetdecomp}
	Given $F$, a set of flags, let $s$ be a type $T$ collection of pins and let $T^\prime$ be a type containing $T$, i.e. $T^\prime\supset T$.
	The pinned set $P_T(s)$ is partitioned into some number, say $m$, of pinned sets, each characterized by a type $T^\prime$ collections of pins containing $s$, i.e. $s^\prime_j \supset s$,
	\[P_T(s) = \bigsqcup_{j=1}^mP_{T^\prime}(s^\prime_j).\]
\end{prop}
\begin{IEEEproof}
	Let $s$ be a type $T$ collection of pins and let $T^\prime$ be a type such that $T^\prime\supset T$.
	Define the following set of collections of pins
	\begin{equation}
	S = \Pi_{T^\prime}\left (P_T(s)\right ),
	\end{equation}
	then it is the case that
	\begin{equation}
	P_T(s) = \bigsqcup_{s^\prime\in S} P_{T^\prime}(s^\prime).
	\end{equation}
	Indeed, since $T^\prime\supset T$ and $\forall s^\prime\in S,\,s^\prime\supset s$ we have that $\forall s^\prime\in S,\,P_{T^\prime}(s^\prime)\subset P_T(s)$ showing the right to left inclusion.
	The left to right inclusion follows from the definition of $S$.
	Finally the fact that the union is disjoint follows from Prop.~\ref{prop:interpinset}.
\end{IEEEproof}

\subsection{Definition of a $(x,z)$-pin code}

Equipped with the notions layed out in the previous section, we now construct quantum codes.
They are defined by a choice of flags $F$ and two natural integers $x$ and $z$ which fulfill the condition $x+z \leq D$.
The flags are identified with qubits, so that $n = \vert F\vert$. 
The $X$-stabilizers generators are defined by all the $x$-pinned sets and the $Z$-stabilizers generators by the $z$-pinned sets in the following way.
For each $x$-pinned set we define the Pauli operator acting as the Pauli-$X$ operator on the qubits corresponding to the flags in the pinned set and as the identity on the rest; we then add it to the generating set of the $X$-stabilizers.
We do similarly for the $z$-pinned sets to form the generating set of the $Z$-stabilizers.
In order to ensure the correct commutation relations between the $X$- and $Z$-stabilizers it is sufficient to enforce the following condition on the relation $F$.
\begin{defi}[Pin code relation]
	A $(D+1)$-ary relation, $F$, is a \emph{pin code relation} if all $D$-pinned sets have even cardinality.\label{def:pincoderel}
\end{defi}
This property is sufficient for the overlap between two pinned sets to be even every time they are pinned by a small enough number of pins.
This is summarized in the following proposition.
\begin{prop}[Even overlap]\label{prop:evenoverlap}
	Let $F$ be a $(D+1)$-ary pin code relation as per Definition~\ref{def:pincoderel}.
	Let $x$ and $z$ be two natural integers such that, \[x+z\leq D.\]
	Then, the intersection between any $x$-pinned set and any $z$-pinned set has even cardinality.
\end{prop}
\begin{IEEEproof}
	By Proposition~\ref{prop:interpinset}, the intersection of a $x$-pinned set and a $z$-pinned set is either empty or a pinned set with at most $x+z\leq D$ pins.
	In turn by Proposition~\ref{prop:pinsetdecomp}, the intersection can be partitionned into $D$-pinned sets which all have even cardinality since the relation $F$ is a pin code relation.
\end{IEEEproof}
Using a pin code relation we can therefore define a CSS code as follows.
\begin{defi}[$(x,z)$-pin code]
	Given a pin code relation, $F$, on $(D+1)$ sets and two natural integers $(x,z)\in\mathbb{N}^2$, such that $x+z\leq D$, we define the corresponding $(x,z)$-pin code by associating the elements of $F$ with qubits, all the $x$-pinned sets with $X$-stabilizer generators and all the $z$-pinned sets with $Z$-stabilizer generators.
	By Proposition~\ref{prop:evenoverlap} the defined code is a valid CSS code.
\end{defi}
A first remark is that the choice of $x=0$ (or~$z=0$) is not particularly interesting since in this case there is a single $X$-stabilizer (or $Z$-stabilizer) acting on all the qubits.
A second remark is that in the strict case, where $x+z < D$, the code will contain many low weight logical operators which are naturally identified as gauge operators. 
This is explained in details in Sec.~\ref{sec:gaugePC}.

Table~\ref{tab:summarypincodes} summaries the objects defined and their role to define a quantum code.
The main task is to construct a pin code relation which permits to define the rest.
Before explaining how to construct pin code relations we show that Reed-Muller codes can be generated by pinned-sets of a specific relation, and that quantum color codes are a subclass of pin codes.

\begin{table}[h]
	\centering
	\caption{Summary of the objects defining quantum pin codes.}
	\label{tab:summarypincodes}
	\begin{tabular}{l c | c}
		\multicolumn{2}{c}{Objects} & Quantum code\\
		\hline\\[-.5em]
		Level: & $L_j$ & $-$\\
		Pin: & $p\in L_j$ & $-$\\
		Type: & $T\subset\{0,\ldots,D\}$ & $-$\\
		Collection: & $s\in L_{j_1}\times\cdots\times L_{j_k}$ & $-$\\
		Relation: & $F\subset L_0\times\cdots\times L_D$ & All qubits\\
		Flag: & $f\in F$ & Qubit\\
		$x$-pinned set: &$P_T(s)\subset F\quad (\vert T\vert = x)$ & $X$-stabilizer\\
		$z$-pinned set: &$P_T(s)\subset F\quad (\vert T\vert = z)$ & $Z$-stabilizer\\
	\end{tabular}
\end{table}

\subsection{Relation to Reed-Muller codes}
\label{sub:reedmuller}

Given integers $r$ and $m$, the classical Reed-Muller code, $\mathcal{RM}(r,m)$, is defined as the space generated by the vectors of all evaluations of polynomials over $m$ variables of degree at most $r$:
\begin{eqnarray}
	\mathcal{RM}(r,m) = \Big\{\left (p(x)\right )_{x\in \mathbb{F}_2^m}:\; &p\in\mathbb{F}_2[X_0,\ldots,X_{m-1}],\nonumber\\
	&{\rm deg}(p)\leq r\Big\}.
\end{eqnarray}

We now show that $\mathcal{RM}(r,m)$ is generated by the pinned sets with $r$ pins of a certain pin code relation $F$.
Consider $m$ levels, each containing two pins, 
\begin{equation}
	\forall j\in\{0,\ldots,m-1\},\;L_j = \{0,1\},\label{eq:RMlevels}
\end{equation}
and the complete $m$-ary relation,
\begin{equation}
F_{\mathcal{RM}}^m = L_0\times\cdots\times L_{m-1},\label{eq:RMrelation}
\end{equation}
see also the left of Figure~\ref{fig:EiffelTCC}.
Consider now a $r$-pinned set defined, with type $T=\{j_1, \ldots, j_r\}$, and pins $b=(b_{j_1},\ldots,b_{j_r})$.
One can check that a flag, $f\in F=\mathbb{F}_2^m$, belongs to $P_T(b)$ if and only if the following degree $r$ polynomial, $p_{T,b}$, evaluates to $1$,
\begin{equation}
	p_{T,b} = \prod_{\substack{j\in T\\b_j=1}}X_j\prod_{\substack{k\in T\\b_k=0}}\left (1-X_k\right ).
\end{equation}
So the pinned sets with $r$ pins generate the following code,
\begin{eqnarray}
	\mathcal{P}(r,m) = 
	\Big\langle\left (p_{T,b}(x)\right )_{x\in\mathbb{F}_2^m}:&\;T\subset\{0,\ldots,m-1\},\nonumber\\
	&\vert T\vert = r,\, b\in\mathbb{F}_2^m\Big\rangle.
\end{eqnarray}
Then we just have to check that these generate all polynomial of degree at most $r$.
By definition they generate polynomials constituted of a product of $r$ elements being either $X_j$ or $(1-X_j)$.
Let's suppose, for some $\ell\leq r$, they can generate all such product with only $\ell$ terms.
Then we can contruct all product with only $\ell-1$ terms, $q$, as follows
\begin{equation}
	q = q\cdot(1-X_j) + q\cdot X_j,
\end{equation}
where $X_j$ is a variable that does not appear in $q$.
It follows by induction that they generate all degree at most $r$ polynomials and so
\begin{equation}
	\mathcal{RM}(r,m) = \mathcal{P}(r,m).
\end{equation}

Another way to see this is to use the decomposition property of pinned sets (Proposition~\ref{prop:pinsetdecomp}) and generate the lower degree monomial directly with pinned sets with less pins and decompose these pinned sets into disjoint union of $r$-pinned sets.\\

Classical Reed-Muller codes have been used in several different ways to define quantum codes. 
Some of their useful properties include that the dual of $\mathcal{RM}(r,m)$ is $\mathcal{RM}(m-1-r,m)$ and that if $r\leq s$ then $\mathcal{RM}(r,m)\subseteq\mathcal{RM}(s,m)$.
Some constructions use directly Reed-Muller codes to define the $X$ and $Z$ stabilizers \cite{steaneQuantumReedMullerCodes1999, zhangQuantumReedMullerCodes1997, rengaswamy_optimality_2020}.
In \cite{steaneQuantumReedMullerCodes1999} the $X$-stabilizers are given by $\mathcal{RM}(r,m)$ and the $Z$-stabilizers are given by $\mathcal{RM}(r-1,m)$. 
In addition some of the generators for the $X$-stabilizers can be modified to increase the distance of the code by half while loosing the CSS property.
Before this modification this construction can be directly interpreted as the quantum pin code defined by the relation $F_{\mathcal{RM}}^m$ from \eqref{eq:RMrelation} and by setting $x=r$ and $z=r-1$ with $2r-1\leq m-1$ which ensures that $x+z\leq D=m-1$.
In the papers \cite{zhangQuantumReedMullerCodes1997, rengaswamy_optimality_2020}, the $X$-stabilizers are given by $\mathcal{RM}(r-1, m)$ and the $Z$-stabilizers by $\mathcal{RM}(m-r-1, m)$.
This can be directly interpreted as the quantum pin code defined by the relation $F_{\mathcal{RM}}^m$ from \eqref{eq:RMrelation} and by setting $x=r-1$ and $z=m-r-1$ (then automatically $x+z\leq D= m-1$).
 
Several other works use the shortened versions of Reed-Muller codes \cite{knillThresholdAccuracyQuantum1996, bravyi_universal_2005, sarvepalliNonbinaryQuantumReedMuller2005, bravyi_magic-state_2012, campbellMagicStateDistillationAll2012, landahlComplexInstructionSet2013, andersonFaultTolerantConversionSteane2014}.
The shortened version of the Reed-Muller code consists in removing the constant $1$ from the polynomials and throwing away the evaluation bit $p(0)$ (it is now always $0$).
The idea is that if one where to take $\mathcal{RM}(r, m)$ and $\mathcal{RM}(m-r-1, m)$ as $X$ and $Z$ stabilizers respectively, one would get exactly $0$ logical qubits since these two codes are dual of one another.
Taking the shortened versions one gets a single logical qubit.
This can be interpreted as a quantum pin code defined by the relation $F_{\mathcal{RM}}^m$ from \eqref{eq:RMrelation}, by setting $x$ and $z$ such that $x+z=D=m-1$, and declaring as \emph{free pins} the $1$ in all levels $L_j$ in \eqref{eq:RMlevels}, see Section~\ref{sec:bound}.

\subsection{Relation to quantum color codes}

The formalism and definitions for quantum pin codes can be viewed as a generalization of quantum color codes without boundaries \cite{bombin_topological_2006, bombin_exact_2007, kubica_universal_2015}.
However it is also possible to integrate the notion of boundaries into the general framework of pin codes but we delay these considerations to Sec.~\ref{sec:bound}.
A $D$-dimensional color code is defined by a homogeneous, $D$-dimensional, simplicial complex, triangulating a $D$-manifold, whose vertices are $(D+1)$-colorable.
A $D$-dimensional, simplicial complex is called homogeneous if every simplex of dimension less than $D$ is a face of some $D$-simplex.
The (Poincar\'e) dual of a simplicial complex as defined above is sometimes called a colex \cite{bombin_exact_2007}, it consists in a tessellation where the vertices are $(D+1)$-valent and the $D$-cells are $(D+1)$-colorable.
In the original simplicial complex, the qubits are identified with the $D$-simplices and one chooses two natural integers, $\left (\tilde{x},\tilde{z}\right )\in\mathbb{N}^2$ with $\tilde{x}+\tilde{z}\leq D-2$, to define the $X$- and $Z$-checks using the $\tilde{x}$- and $\tilde{z}$-simplices in the following way.
Each $X$-stabilizer generator is identified with a $\tilde{x}$-simplex which acts as Pauli-$X$ on all $D$-simplices in which it is contained.
Similarly, each $Z$-stabilizer generator operate on all the $D$-simplices as Pauli-$Z$ in which the corresponding $\tilde{z}$-simplex is contained, respectively.

This definition can be restated in the language of pin codes: the $D+1$ levels, $L_0,\ldots,L_D$, are indexed by the $D+1$ colors and each level contains the vertices of a given color.
The flags, $F$, are defined using the $D$-simplices (each containing $D+1$ vertices); this defines a relation, $F\subset L_0\times\cdots\times L_D$ thanks to the colorability condition as each $D$-simplex will not contain two vertices of the same color.
One can further check that the relation $F$ is a pin code relation as stated in Def.~\ref{def:pincoderel}.

\begin{prop}{($D$-simplices form a pin code relation).}
	Let $\mathcal{T}$ be a homogeneous $D$-dimensional simplicial complex whose vertices are $(D+1)$-colorable, triangulating a $D$-manifold.
	Given a $(D+1)$-coloration with colors numbered from $0$ to $D$, define $L_j$ as the set of vertices of color number $j$ in $\mathcal{T}$.
	The relation $F\subset L_0\times\cdots\times L_D$ representing the set of $D$-simplices is a pin code relation.
\end{prop}
\begin{IEEEproof}
	$D$-pinned sets correspond to $(D-1)$-simplices which are contained in exactly two $D$-simplices since the simplicial-complex triangulates a $D$-manifold without boundaries.
	This shows that it is a pin code relation.
\end{IEEEproof}
Then subsets of the $(D+1)$ colors correspond to types and any $k$-simplex corresponds directly to a collection of pins of type given by the $(k+1)$ different colors of the $(k+1)$ vertices of the $k$-simplex.
The corresponding $(k+1)$-pinned set contains all the $D$-simplices containing the original $k$-simplex.
As such all the non-empty $(k+1)$-pinned sets are given by all the collection of pins and type corresponding to all the $k$-simplices. 
With these consideration, we see that choosing $x=\tilde{x}+1$ and $z = \tilde{z}+1$, the corresponding $(x,z)$-pin code is the same as the original color code.

\begin{prop}{(Quantum color codes are quantum pin codes).}
	Let $\mathcal{T}$ be a homogeneous $D$-dimensional simplicial complex whose vertices are $(D+1)$-colorable, triangulating a $D$-manifold.
	For integers $x$ and $z$ such that $x+z \leq D-2$, the $(x-1, z-1)$-color code defined on $\mathcal{T}$ is equal to the $(x,z)$-pin code defined by the relation on the colored vertices given by the set of $D$-simplices.
\end{prop}
\begin{IEEEproof}
	By inspection, using Table~\ref{tab:colorcodespincodes}, one checks that the qubits and stabilizers exactly coincide in the definition of the $(x-1, z-1)$-color code and the $(x,z)$-pin code.
\end{IEEEproof}

\begin{table}[h]
	\centering
		\renewcommand{\arraystretch}{1.3}
	\caption{Correspondence between color codes and quantum pin codes.}
	\label{tab:colorcodespincodes}
	\begin{tabular}{|p{10em}  | l c|}
		\hline
		\centering Color code & \multicolumn{2}{c|}{Quantum pin code}\\
		\hline
		\centering Set of vertices of color $\!j$ & Level: & $L_j$\\\hline
		\centering Vertex of color $j$ & Pin: & $p\in L_j$\\\hline
		\centering Set of $D$-simplices & Relation: & $F\subset L_0\times\cdots\times L_D$\\\hline
		\centering A given $D$-simplex & Flag: & $f\in F$\\\hline
		\centering A given $k$-simplex & Collection: & $s\in L_{j_0}\times\cdots\times L_{j_k}$\\\hline
		\centering All $D$-simplices containing a given $k$-simplex & $(k+1)$-pinned set: & $P_T(s)\subset F$\\\hline
	\end{tabular}
\end{table}

An example for $D=2$, based on the hexagonal color code, is shown in Figure~\ref{fig:hexex}.
To summarize: to go from color codes to pin codes one just forgets the geometry, keeping only the $(D+1)$-ary relation given by the $(D+1)$-colored $D$-simplices.

Importantly, pin codes are more general, as there are pin code relations which are not derived from these specific simplicial complexes.
In the next section after recalling a concrete color code construction we give two more general constructions of pin code relations.

\subsection{Constructing pin codes}
\begin{figure}[h]
	\centering
	\subfloat[]{\includegraphics[width=.49\linewidth]{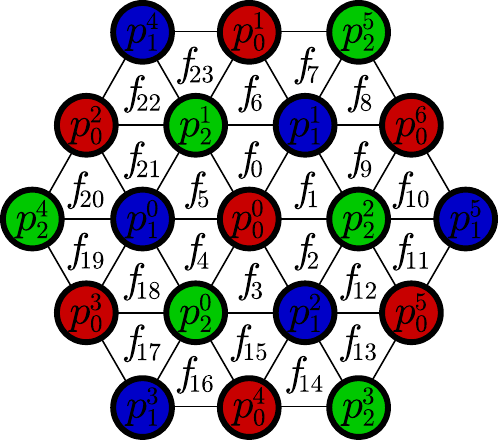}%
		\label{fig:hexex}}\hfil
	\subfloat[]{\includegraphics[width=.49\linewidth]{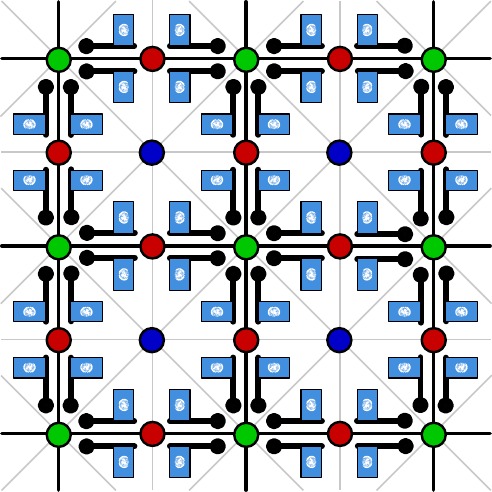}%
		\label{fig:flags}}
	\caption{(a) An example of flag set based on the triangular lattice: $F =\left  \{f_0,\ldots,f_{23}\right \}$ from $D+1=3$ levels: $L_0 = \left \{p_0^0,\ldots,p_0^3\right \}$ in red, $L_1 = \left \{p_1^0,\ldots,p_1^3\right \}$ in blue, $L_2 = \left \{p_2^0,\ldots,p_2^3\right \}$ in green.
	The figure wraps around on itself according to the arrows so that there are no boundaries, and the surface obtained is a torus.
	(b) Schematic representation of the flags of the square lattice and the corresponding pins.
	Each triple of incident vertex (green pin), edge (red pin) and face (blue pin) is a flag.
	They are symbolized in the picture as actual flags put closest to the elements in the triple (vertex, edge, face) they stand for. 
	The corresponding color code is the well known 4.8.8 color code: there are eight flags around each vertex, four around each edge and eight around each face.}
	\label{fig:hexflags}
\end{figure}

\subsubsection{Color codes from tilings}\label{sec:construction_tilings}
In \cite{bombin_exact_2007}, the authors explain how to obtain a colex from any tiled $D$-manifold.
The idea is to successively inflate the $(D-1)$-cells, $(D-2)$-cells, \ldots, $0$-cells into $D$-cells.
The dual of the tiling obtained is then a $(D+1)$-colorable triangulation of the $D$-manifold, see Appendix~A of \cite{bombin_exact_2007}.
This can also be understood directly, without inflating the cells, as follows:
Separate all the cells into $(D+1)$ sets, $L_0,\ldots,L_D$, according to their dimensions, i.e. $L_j$ contains all the $j$-cells.
We can now define a $(D+1)$-ary relation on the cells via the incidence relation.
Two cells of different dimension are incident if and only if one is a sub-cell of the other.
An element of this relation, i.e. a $(D+1)$-tuple containing a $0$-cell (a vertex), a $1$-cell (an edge), etc\ldots up to a $D$-cell, is called a flag.
See for example Fig.~\ref{fig:flags} for a representation of the flags of the square lattice.
The flag relation obtained this way is the same as the one after going through the inflating procedure and it is a pin code relation.

\begin{figure}
	\centering
	\subfloat[]{{\raisebox{4.5em}{\includegraphics[width=.11\linewidth]{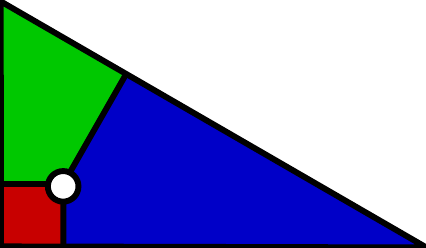}}}%
		\label{fig:wythoff_triangle}}\hfil
	\subfloat[]{\includegraphics[width=.42\linewidth]{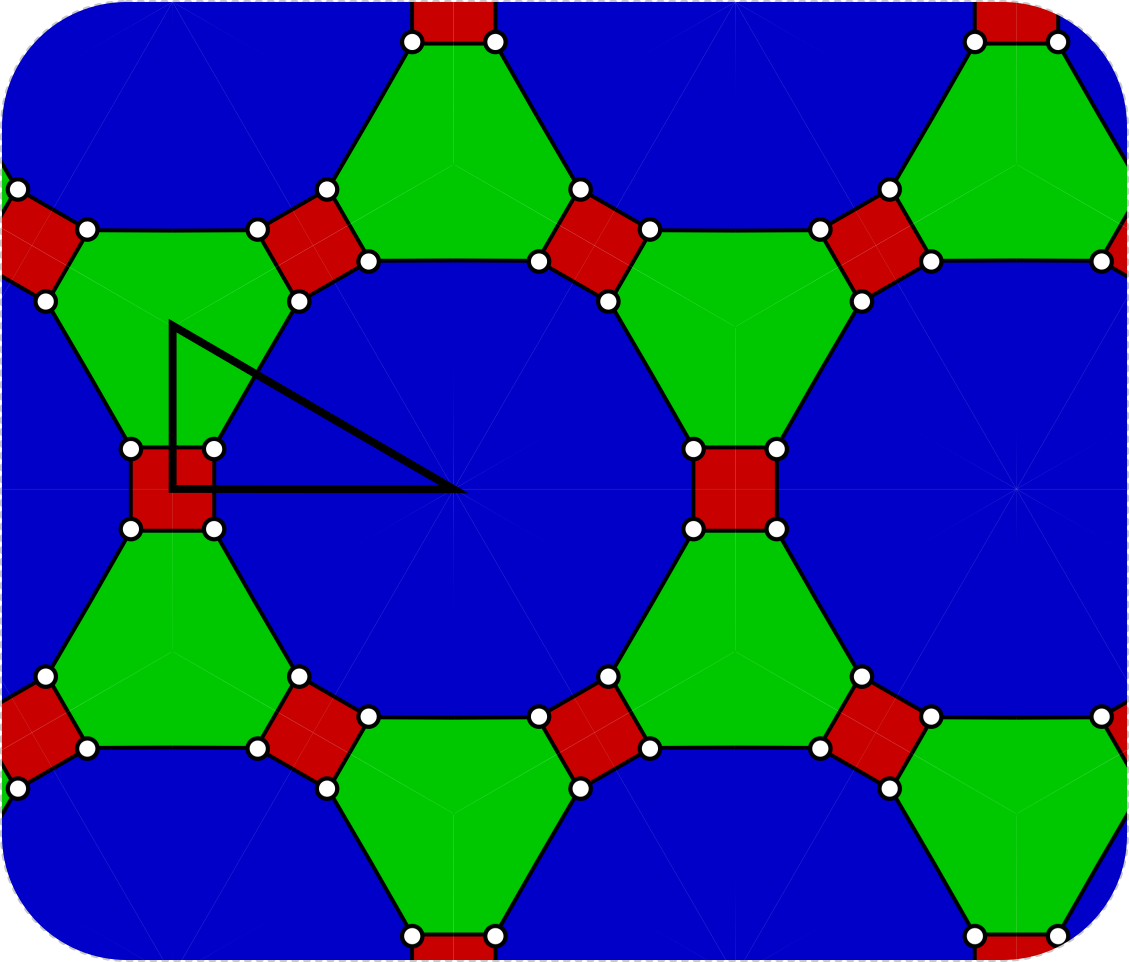}%
		\label{fig:wythoff_lattice}}\hfil
	\subfloat[]{\includegraphics[width=.22\linewidth]{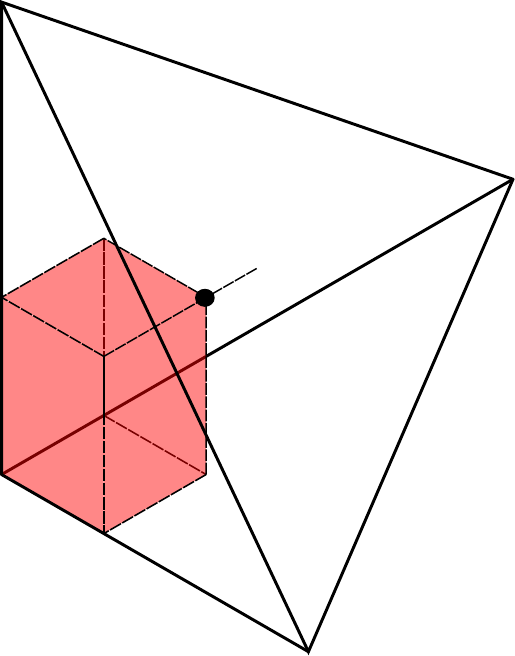}%
		\label{fig:wythoff3D}}
	\caption{(a) The fundamental triangle of the Wythoff construction.  A vertex (white cirlce) is placed in the middle and three edges are drawn to the boundary.  The three regions are colored red, green and blue.
		(b) Reflecting the fundamental triangle along its sides creates a three colored lattice. In this case the 4.6.12 lattice. Note that we can obtain the red/blue/green shrunk lattice by moving the vertex into the red/blue/green corner, see Sec.~\ref{sec:unfold}.
		(c). The Wythoff construction can be generalized to higher dimensions as well. This example shows the vertex put into the middle of a 3D simplex, a tetrahedron, and one of the four corner cells colored in red (see main text for more information).}
	\label{fig:wythoff_construction}
\end{figure}

A similar way to construct $(D+1)$-colorable tessellations in~$D$ dimensions directly is to use the Wythoff construction~\cite{coxeter1973regular}.
The construction is quite general, but for simplicity, let us start on the 2D euclidean plane.
Consider a right-angled triangle and draw a point into its interior. From this point we draw three lines, each intersecting a boundary edge in a right angle (see Fig.~\ref{fig:wythoff_triangle}).
This creates three regions in the triangle which we assign three different colors.
We can now reflect the triangle along its boundary edges.
The internal points of the original and the reflected triangles become the vertices of a uniform tiling. The faces of the tiling are colored by the three colors and by construction no two faces of the same color are adjacent (see Fig.~\ref{fig:wythoff_lattice}).
If the angles of the triangle are $2\pi/r$, $2\pi/s$ and $2\pi/l$ then the result will be a $r$.$s$.$l$-tiling, meaning that the three faces around a vertex will have $r$, $s$ and $l$ number of sides.

This idea readily generalizes to higher dimensions by placing a vertex into a $D$-dimensional simplex and drawing lines to the mid-point of the $D-1$-dimensional faces of the simplex (see Fig.~\ref{fig:wythoff3D} for the case $D=3$).
The faces of the simplex are simplices themselves, so this process can be iterated until $D=2$.
Reflecting along the faces of the $D$-simplex gives rise to a uniform tiling of the $D$-dimensional space with $D$-cells being colored by $D+1$ colors and no two cells of the same color sharing a $D-1$-dimensional face.
The number of vertices of the $D$-cells is then determined by the orbit of the reflections along all but one of the sides.

The color codes from regular tilings can be obtained this way, for example, in 2D Euclidean space, the hexagonal, 4.8.8. or 4.6.12. color codes or more generally both Euclidean and hyperbolic tilings in any dimension.

In order to obtain a finite code from the infinite tessellation we can either introduce boundaries or by introducing periodic boundary conditions. The details will depend on the individual lattice.

\subsubsection{Coxeter group approach}

The Wythoff construction of the previous section can be generalized in the language of group theory.
The reflections used to tile the space form groups called \emph{Coxeter groups}.
The fundamental triangles can be seen as group elements and reflections about their sides are generators of the group. 
In this section we will explain the more general approach based on these groups or more generally on finite groups which are generated by elements with even order.
A Coxeter group is a finitely presented group with reflections as generators, denoted as
\begin{equation}
G = \langle a_0,\ldots,a_D\vert a_0^2 = \cdots = a_D^2 = \left (a_ia_j\right )^{k_{ij}} = r_k = \cdots = 1\rangle,
\end{equation}
where the $k_{ij}$ are integers defining relations between the generators, $r_k$ are additional relations between generators and~$1$ is the trivial element~\cite{coxeter_regular_1973, davis_geometry_2007}.
Define the subgroups, $H_j$ for \mbox{$j\in\{0,\ldots,D\}$}, as
\begin{equation}
H_j = \big\langle\{a_0,\ldots,a_D\}\setminus\{a_j\}\big\rangle.
\end{equation}
Define the levels, $L_j$, as the sets of left cosets for each $H_j$, i.e.
\begin{equation}
L_j = \left \{gH_j \mid g\in G\right \}.
\end{equation}
The cosets of a subgroup always form a partition of the full group.
So for every $j\in\{0,\ldots,D\}$, a group element $g\in G$ uniquely defines a coset $p_j\in L_j$ such that $g\in p_j$.
Hence, each group element defines a $(D+1)$-tuple of cosets, $(p_0,\ldots,p_D)\in L_0\times\cdots\times L_D$.
Taking the set of all such tuples defines a $(D+1)$-ary relation on the cosets,  $F\subset L_0\times\cdots\times L_D$, which is a pin code relation.
The fact that this $F$ is a pin code relation can be verified by the following argument.
A $k$-pinned set here correspond to the intersection of $k$ different cosets with respect to $k$ different subgroups $H_j$.
It always holds that the intersection of several cosets is either empty or is a coset with respect to the intersection of the subgroups of the original cosets.
Hence non-empty $k$-pinned sets are cosets with respect to a subgroup, $H_{j_1,\ldots,j_k}$,
\begin{equation}
H_{j_1,\ldots,j_k} = \bigcap_{i=1}^kH_{j_i}.
\end{equation}
Each subgroup $H_{j_i}$ is generated by all generators of $G$ except one, this means that $H_{j_1,\ldots,j_k}$ contains at least a subgroup generated by $D-k+1$ of the generators,
\begin{equation}
H_{j_1,\ldots,j_k} \supseteq \big\langle\{a_0,\ldots,a_D\}\setminus\{a_{j_1},\ldots,a_{j_k}\}\big\rangle.\label{eq:subgroupinter}
\end{equation}
In particular, the $D$-pinned sets are cosets with respect to a subgroup that contains $\langle a_j\rangle$ for some $j$.
We note that $\langle a_j\rangle$ has order two, since $a_j$ is a reflection, and thus all $D$-pinned sets have even order.
In well behaved cases the containment in \eqref{eq:subgroupinter} will actually be an equality, but this is not guaranteed due to relations between generators that may exist.

In the case where the Coxeter group describes the symmetries of a tiling this construction is equivalent to the Wythoff construction described above.
But one also obtains more general pin codes when considering Coxeter groups not defining tilings or more general finite groups with generators of even order.

Note that it is generally not possible to directly obtain the properties of the quantum pin code derived from a Coxeter group from the group's representation and the choice of subgroups~$H_i$.

\begin{table}[h]
	\centering
		\renewcommand{\arraystretch}{1.3}
	\caption{Correspondence between Coxeter groups and quantum pin codes.}
	\label{tab:coxetergrouppincodes}
	\begin{tabular}{|p{10em}  | l c|}
		\hline
		\centering Coxeter group & \multicolumn{2}{c|}{Quantum pin code}\\\hline
		\centering All left cosets of $H_j= \big\langle\{a_0,\ldots,a_D\}\setminus\{a_j\}\big\rangle$ & Level: & $L_j$\\\hline
		\centering A given left coset $gH_j$ & Pin: & $p\in L_j$\\\hline
		\centering All group elements $G$ & Relation: & $F\subset L_0\times\cdots\times L_D$\\\hline
		\centering A group element $g\in G$ & Flag: & $f\in F$\\\hline
		\centering $k+1$ left cosets $(g_0H_{j_0},\ldots,g_kH_{j_k})$ & Collection: & $s\in L_{j_0}\times\cdots\times L_{j_k}$\\\hline
		\centering Elements of a left coset $gH_{j_0,\ldots,j_k}$ & $(k+1)$-pinned set: & $P_T(s)\subset F$\\\hline
	\end{tabular}
\end{table}

In Sec.~\ref{sec:hyperbolic} we explore in more details the construction of pin codes from 3D hyperbolic Coxeter groups and give some explicit examples.\\

\subsubsection{Chain complex approach}
\label{sec:chaincomplexeapproach}

An other way of obtaining a pin code relation is from $\mathbb{F}_2$ chain complexes of length $D+1$.
For our purposes, these algebraic objects are composed of $(D+1)$ finite-dimensional vector spaces over $\mathbb{F}_2$, say $\mathcal{C}_0,\ldots,\mathcal{C}_D$, together with $D$ linear maps called boundary maps, $\partial_j:\mathcal{C}_j\rightarrow\mathcal{C}_{j-1}$, which are such that
\begin{equation}
\forall j\in\{0,\ldots,D-1\},\,\partial_j\circ\partial_{j+1} = 0.\label{eq:boundarymap}
\end{equation}
For example the tiling of a $D$-manifold can be seen as a chain complex, taking the $j$-cells as a basis for the $\mathcal{C}_j$ vector space and the natural boundary map.
We have shown how to get a pin code relation from such a tiling by taking its flags, but it can as well be obtained from any $\mathbb{F}_2$ chain complex. For more background see~\cite{rotman2008introduction}.

The construction works as follows: choose a basis set $L_j$ for each vector space $\mathcal{C}_j$.
The $L_j$ basis sets are the levels and the basis elements the pins.
Then use the boundary map,~$\partial$, to define binary relations, $R_{j,j+1}\subset L_{j}\times L_{j+1}$, where $(p_{j},p_{j+1})\in R_{j, j+1}$ if $p_j$ appears in the decomposition of  $\partial\left (p_{j+1}\right )$ over the basis set $L_{j}$.
Then the relation $F\subset L_0\times\cdots\times L_D$ is defined as follows
\begin{equation}
F = \left \{(p_0,\ldots,p_D)\mid \forall j,\;\left (p_j,p_{j+1}\right )\in R_{j,j+1}\right \}.
\end{equation}
The relation $F$ obtained like this is almost a pin code relation.
All the pinned sets of type $t = \{0,\ldots,D\}\setminus \{j\}$ with $0<j<D$ have even cardinality since their size is given by the number of paths between the pin $p_{j+1}$ and the pin $p_{j-1}$ which has to be even by the property of the boundary map $\partial$ given in \eqref{eq:boundarymap}.
For pinned sets of type $T = \{1,\ldots,D\}$ or $T = \{0,\ldots,D-1\}$ it is not generally the case that they have even cardinality.
Although this can be easily fixed by adding at most two pins:
the idea is then to add one rank-$0$ pin, $b_0$, in the level $L_0$ and add all pairs $(b_0,c^\star)$ such that
\begin{equation}
\left \vert\{p_0 \mid (p_0,c^\star)\in R_{0,1}\}\right \vert = 1\pmod2,
\end{equation}
to the new relation $R_{0,1}$.
Then do the same for the level $D$, adding $b_D$ in $L_D$.
After this modification the resulting flag relation $F$ is a pin code relation.

\begin{table}[h]
	\centering
		\renewcommand{\arraystretch}{1.3}
	\caption{Correspondence between chain complexes and quantum pin codes.}
	\label{tab:chaincomplexpincodes}
	\begin{tabular}{|p{10em}  | l c|}
			\hline
		\centering Chain Complex & \multicolumn{2}{c|}{Quantum pin code}\\
		\hline
		\centering Basis set of $\mathcal{C}_j$& Level: & $L_j$\\\hline
		\centering A basis element $p_j\in\mathcal{C}_j$ & Pin: & $p\in L_j$\\\hline
		\centering All basis elements $(p_0,\ldots,p_D)$ such that $p_{j}\in\partial p_{j+1}$ & Relation: & $F\subset L_0\times\cdots\times L_D$\\\hline
		\centering A tuple $(p_0,\ldots,p_D)$ such that $p_{j}\in\partial p_{j+1}$ & Flag: & $f\in F$\\\hline
		\centering $k+1$ basis elements from different levels $\left (p_{j_0}, \ldots, p_{j_k}\right )$ & Collection: & $s\in L_{j_0}\times\cdots\times L_{j_k}$\\\hline
		\centering All tuples $\left (q_0,\ldots,q_D\right )\in F$ such that $\left (q_{j_0},\ldots,q_{j_k}\right ) = \left (p_{j_0},\ldots,p_{j_k}\right )$& $(k+1)$-pinned set: & $P_T(s)\subset F$\\\hline
	\end{tabular}
\end{table}

Note that this way of obtaining a quantum code from any $\mathbb{F}_2$ chain complex is fundamentally different from the usual homological code construction.
In the homological code construction one chooses one of the levels, say $L_j$, and identifies its elements with qubits.
Then the $Z$-stabilizer generators are given by the boundary of the elements in $L_{j+1}$ and the $X$-stabilizer generators by the coboundary of the elements in~$L_{j-1}$.
These are different from the flags and pinned sets used to define a pin code.

In Sec.~\ref{sec:hp} we give some explicit pin codes constructed from chain complexes.

\subsection{Remarks}
\label{sec:splitable}
While some flag relations $F$ obtained from Coxeter groups can be equivalently viewed as coming from some $\mathbb{F}_2$ chain complex, the converse does not necessarily hold.
Indeed not every multi-ary relation can be decomposed into a sequence of binary relations, the hexagonal lattice depicted in Fig.~\ref{fig:hexex} is an example of such a relation which cannot be decomposed this way.
The other way around, not all flag relations obtained from a $\mathbb{F}_2$ chain complex can be seen as coming from a Coxeter groups as in general they would lack the regular structure required. 

Depending on the pin code relation, $F$, it can happen that some pinned sets can in fact be safely split when defining the stabilizers.
That is to say, one can separate them into several disjoint sets of flags defining each an independent stabilizer still commuting with the rest of the stabilizers.
For example this is the case for Coxeter groups for which \eqref{eq:subgroupinter} is strict, i.e.
\begin{align}
\big\langle\Set{a_{i_1},\ldots,a_{i_s}}\big\rangle\cap\big\langle\Set{a_{j_1},\ldots,a_{j_t}}\big\rangle \supsetneq\nonumber\\ \big\langle\Set{a_{i_1},\ldots,a_{i_s}}\cap\Set{a_{j_1},\ldots,a_{j_t}}\big\rangle.\label{eq:c-group}
\end{align}
In this case the cosets with respect to the first group can be further split into cosets with respect to the second one without harming the commutation relations.
Groups generated by reflections for which \eqref{eq:subgroupinter} is always an equality are called C-groups \cite{mcmullen_abstract_2002}.
If the stabilizers are still defined as whole pinned sets, in cases where they could be split, then these smaller sets of qubits would be logical operators which would be detrimental to the overall distance of the code.

\section{Transversal gates and magic state distillation}
\label{sec:multiorthogonality}

In this section we present independently of pin codes what structure is desirable for CSS codes to admit transversal phase gates of different levels of the Clifford hierarchy.
The presentation here is close in spirit to that of \cite{campbell_unified_2017, haah_codes_2018}.
It is included here to fix terminology and to be self-contained.

Given $\ell$ binary row vectors, $\bs{v}^1, \ldots, \bs{v}^\ell\in\mathbb{F}_2^n$, we denote their element-wise product as, $\bs{v}^1\wedge\cdots\wedge\bs{v}^\ell$, its $j^\text{th}$ entry is given by
\begin{equation}
\left [\bigwedge_{m=1}^\ell\bs{v}^m\right ]_j =\left [\bs{v}^1\wedge\bs{v}^2\wedge\cdots\wedge\bs{v}^\ell\right ]_j = v_j^1v_j^2\cdots v^p_j.
\end{equation}
The Hamming weight of a binary vector $\bs{v}$ is denoted as $\vert\bs{v}\vert$, it is given by the sum of its entries.
We also define the notions of multi-even and multi-orthogonal spaces:
\begin{defi}[Multi-even space]\label{def:multieven}
	Given an integer, $\ell\in\mathbb{N}$, a subspace $\mathcal{C}\subset\mathbb{F}_2^n$, is called $\ell$-even if all vectors in $\mathcal{C}$ have Hamming weight divisible by $2^\ell$:
	\[\forall\bs{v}\in\mathcal{C},\;\left \vert\bs{v}\right \vert = 0 \pmod{2^\ell}.\]
\end{defi}

\begin{prop}[Characterization of multi-even spaces]\label{prop:equmultieven}
	Given $\ell\in\mathbb{N}$, a subspace $\mathcal{C}\subset\mathbb{F}_2^n$ is $\ell$-even if and only if for any integer $s\in\{1,\ldots,\ell\}$ and any $s$-tuple of vectors, $(\bs{v}^1,\ldots,\bs{v}^s)\in\mathcal{C}^s$, it holds that
	\begin{equation}
	\left|\bs{v}^1\wedge\cdots\wedge\bs{v}^s\right| = 0 \pmod{2^{\ell - s +1}}.\label{eq:multieven}
\end{equation}
\end{prop}
The proof is deferred to Appendix~\ref{sec:proofmultiorthogonality}, it makes use of the following convenient lemma for converting binary addition to regular integer addition.
\begin{lem}[Binary addition and integer addition]\label{lem:idbinadd}
	Denote binary addition with $\oplus$ and integer addition with $+$ or $\sum$.
	Given $r$ binary vectors $\bs{w}^1,\ldots,\bs{w}^r\in\mathbb{F}_2^n$ it holds that
	\begin{equation}
	\bigoplus_{m=1}^r\bs{w}^m = \sum_{s=1}^r(-2)^{s-1}\sum_{1\leq m_1<\cdots<m_s\leq r}\,\bigwedge_{i=1}^s\bs{w}^{m_i}.\label{eq:idbinadd}
	\end{equation}
	Similarly for their Hamming weights
	\begin{equation}
	\left \vert\bigoplus_{m=1}^r\bs{w}^m\right \vert = \sum_{s=1}^r(-2)^{s-1}\sum_{1\leq m_1<\cdots<m_s\leq r}\,\left \vert\bigwedge_{i=1}^s\bs{w}^{m_i}\right \vert.\label{eq:idbinaddhweight}
	\end{equation}
\end{lem}
The proof of this lemma is deferred to Appendix~\ref{sec:proofmultiorthogonality} where we also give an illustrative example.
The alternative characterization given in Prop.~\ref{prop:equmultieven} for a multi-even space, as per Def.~\ref{def:multieven}, shows that some strong divisibility conditions on the overlap between vectors of the space have to be imposed.

It is possible to relax these conditions and define what is called a multi-orthogonal space.

\begin{defi}[Multi-orthogonal space]\label{def:multiorth}
	Given an integer, $\ell\in\mathbb{N}$, a subspace $\mathcal{C}\subset\mathbb{F}_2^n$, is called $\ell$-orthogonal if for any $\ell$-tuple of vectors, $\left (\bs{v}^1,\ldots,\bs{v}^\ell\right )\in\mathcal{C}^\ell$,
	\begin{equation}
	\left|\bs{v}^1\wedge\cdots\wedge\bs{v}^\ell\right| = 0 \pmod 2.
	\label{eq:multiorth}
	\end{equation}
\end{defi}

These two characterizations, Prop.~\ref{prop:equmultieven} and Def.~\ref{def:multiorth}, are expressed as conditions on any tuple of vectors from the space.
Although it is strong enough to verify the conditions only on a basis of the space, which is proven in Appendix~\ref{sec:proofmultiorthogonality}.

\begin{prop}[multi-even/orthogonal space verification on a basis]\label{prop:basisverif}
	It is necessary and sufficient to verify \eqref{eq:multieven} or \eqref{eq:multiorth} on a basis to ensure that a space is multi-even or multi-orthogonal respectively.
\end{prop}

The single-qubit phase gates are denoted as
\begin{equation}
R_\ell = \begin{pmatrix}
1 & 0\\
0 &\omega_\ell
\end{pmatrix},\qquad \omega_\ell = \e^{i\frac{2\pi}{2^\ell}}.
\end{equation}
For instance $R_1 = Z$, $R_2=S$ and $R_3 = T$ in the usual notations.

\subsection{Weighted polynomials and transversal gates}

Given $k$ qubits and some integer $\ell$, we consider quantum gates, $U_{F_\ell}$, acting diagonally on the computational basis, such that for $\bs{x}\in\mathbb{F}_2^k$,
\begin{equation}
U_{F_\ell}\ket{\bs{x}} = \omega_\ell^{{F_\ell}(\bs{x})}\ket{\bs{x}},
\end{equation}
where ${F_\ell}$ is a so-called weighted polynomial of the form
\begin{align}
{F_\ell}(\bs{x})&= \sum_{s=1}^\ell2^{s-1}\sum_{1\leq m_1<\cdots<m_s\leq \ell}\alpha_{m_1\ldots m_s}\cdot x_{m_1}\cdots x_{m_s},
\end{align}
with coefficients $\alpha_{m_1\ldots m_s}$ in $\mathbb{F}_{2^\ell}$.
Any such gate $U_{F_\ell}$ belongs to the $\ell\text{th}$ level of the Clifford hierarchy \cite{Cui_diagonal_2017}.
Examples, and generating set for $\ell=3$, are given in Table~\ref{tab:weightedpoly}.
\begin{table}[h!]
	\centering
	\renewcommand{\arraystretch}{1.3}
	\setlength{\tabcolsep}{2.5pt}
	\caption{Weighted polynomials corresponding to some well known gate, for $\ell=3$, arranged according number of qubits involved and level of the Clifford hierarchy.}
	\label{tab:weightedpoly}
	\begin{tabular}{| c | c | c | c |}
		\hline
		\!$U_{F_3}\leftrightarrow F_3(\bs{x})$ & 1st level                 & 2nd level                           & 3rd level                              \\ \hline
		1 qubit                              & \!$Z\otimes\id\otimes\id \leftrightarrow  4x_1\!$ & $S\otimes\id\otimes\id \leftrightarrow 2x_1$            & $T\otimes\id\otimes\id \leftrightarrow x_1$                \\
		2 qubits                             & -                         & ${\rm C}Z\otimes\id \leftrightarrow  4x_1x_2$ & ${\rm C}S\otimes\id \leftrightarrow 2x_1x_2$     \\
		3 qubits                             & -                          & -                                    & ${\rm CC}Z \leftrightarrow 4x_1x_2x_3$ \\ \hline
	\end{tabular}
\end{table}

The goal is to implement such a gate on the logical level of a quantum error correcting code by the transversal application of some phase gates.
A transversal application of the single-qubit phase gate $R_\ell$ means that each qubit is acted on with $R_\ell$, which amounts to applying the tensor product $R_\ell^{\otimes n} = R_\ell\otimes\cdots\otimes R_\ell$ on the full system.
Given an $\llbracket n, k, d\rrbracket $ quantum CSS code, define $G$ as the $r\times n$ matrix whose rows describe a generating set of the $X$-stabilizers of the code, and define $L$ as the $k\times n$ matrix whose rows describe a basis for the $X$-logical operators.
The code state in this basis corresponding to $\bs{x}\in\mathbb{F}_2^k$ can then be expressed as
\begin{equation}
\ket{\overline{\bs{x}}} =\frac{1}{\sqrt{2^{r}}} \sum_{\bs{y}\in\mathbb{F}_2^{r}}\ket{\bs{x}L\oplus \bs{y}G}.
\end{equation}

Applying transversally the gate $R_\ell$ on this code state, $\ket{\overline{\bs{x}}}$, yields
\begin{equation}
R_\ell^{\otimes n}\ket{\overline{\bs{x}}} = \frac{1}{\sqrt{2^{r}}} \sum_{\bs{y}\in\mathbb{F}_2^{r}}\omega_\ell^{\vert \bs{x}L\oplus \bs{y}G\vert}\ket{\bs{x}L\oplus \bs{y}G}.\label{eq:transvR}
\end{equation}
Using \eqref{eq:idbinadd}, the power of $\omega_\ell$ above can be rewritten in three parts as 
\begin{equation}
\vert\bs{x}L\oplus\bs{y}G\vert = F_\ell(\bs{x}) + F^\prime_\ell(\bs{y}) + F^{\prime\prime}_\ell(\bs{x},\bs{y}),
\end{equation}
where we defined
\begin{align}
F_\ell(\bs{x}) &= \left \vert\bs{x}L\right \vert\label{eq:Flx} \\
F^\prime_\ell(\bs{y}) &= \left \vert\bs{y}G\right \vert \label{eq:Fly}\\
F^{\prime\prime}_\ell(\bs{x}, \bs{y}) &= -2\left \vert\bs{x}L\wedge\bs{y}G\right \vert \label{eq:Flxy}
\end{align}

Using again \eqref{eq:idbinadd}, one can express these three parts as weighted polynomials whose coefficients are given by the different overlap between $X$-logical operator generators, between $X$-stabilizer generators or between both, see Appendix~\ref{sec:multiorth}.

Provided that it is possible to cancel the action of $F^\prime_\ell(\bs{y})$ and $F^{\prime\prime}_\ell(\bs{x},\bs{y})$ then the resulting operation would correspond to the gate $U_{F_\ell}$ on the logical qubits of the code.
The following two properties of CSS codes are designed to get rid of these two unwanted parts.
\begin{prop}[Exact transversality]\label{prop:transv}
	Let $\mathcal{C}$ be a CSS code, given an integer $\ell$, the code $\mathcal{C}$ allows for the transversal application of $R_\ell$ if the following conditions hold:
	\begin{enumerate}[label=(\roman*)]
		\item The $X$-stabilizers form an $\ell$-even space.
		\item Element-wise products of a $X$-logical operator and a $X$-stabilizer always have Hamming weight divisible by $2^{\ell-1}$.
	\end{enumerate}
	The gate performed at the logical level is then given by the weighted polynomial in \eqref{eq:Flx}.
\end{prop}
Indeed the two conditions above exactly give
\begin{equation}
F^{\prime}_\ell(\bs{y}) = F^{\prime\prime}_\ell(\bs{x},\bs{y}) = 0 \pmod{2^\ell},
\end{equation}
which precisely enforce that the actions of $F^\prime_\ell(\bs{y})$ and $F^{\prime\prime}_\ell(\bs{x},\bs{y})$ are trivial.
We can also settle for a weaker condition under which the unwanted part is not trivial but belong to the $(\ell-1)$th level of the Clifford hierarchy.
\begin{prop}[Quasi-transversality]\label{prop:quasitransv}
	Let $\mathcal{C}$ be a $\llbracket n,k,d\rrbracket $ CSS code, with $r\times n$ generating matrix, $G$, for its $X$-stabilizers and $k\times n$ generating matrix, $L$, for its $X$-logical operators.
	Given an integer $\ell$, the code $\mathcal{C}$ allows for the transversal application of $R_\ell$ up to a $(\ell-1)$th-level Clifford correction if the following conditions hold:
	\begin{enumerate}[label=(\roman*)]
		\item The $X$-stabilizers form an $\ell$-orthogonal space.
		\item For any choice of $s\geq1$ $X$-logical operators and $t\geq1$ $X$-stabilizers with $s+t\leq \ell$
		\begin{align*}
		\left\vert\bigwedge_{i=1}^s \bs{L}^{m_i}\bigwedge_{j=1}^t\bs{G}^{n_j}\right \vert = 0 &\pmod{2}.
		\end{align*}
	\end{enumerate}
The gate performed at the logical level after correction is then given by the weighted polynomial in \eqref{eq:Flx}.
\end{prop}
Indeed, under this condition it follows that, see Appendix~\ref{sec:multiorth}
\begin{equation}
\omega_\ell^{F^\prime_\ell(\bs{y})+F^{\prime\prime}_\ell(\bs{x},\bs{y})} = \omega_\ell^{2\tilde{F}_{\ell - 1}(\bs{x},\bs{y})} = \omega_{\ell-1}^{\tilde{F}_{\ell - 1}(\bs{x},\bs{y})},\label{eq:l-1correction}
\end{equation}
where $\tilde{F}_{\ell-1}$ is a properly weighted polynomial which defines a $(\ell -1)$th-level Clifford correction to be applied.
The exact correction is given by the conjugation of $U_{\tilde{F}_{\ell-1}}$ by a decoding circuit for $\mathcal{C}$, see \cite{campbell_unified_2017}.
Note that we can also define intermediate conditions so that the correction belongs to the $(\ell - q)$th level of the Clifford hierarchy for some $1\leq q \leq\ell$.
Here we have just defined the two extreme ones, for which the correction belongs either to the $0$th level or the $(\ell-1)$th level of the Clifford hierarchy.
Propositions~\ref{prop:transv} and \ref{prop:quasitransv} where stated in slightly different form in \cite{campbell_unified_2017}.

\subsection{Magic-state distillation}

Given a code which exhibits exact transversality or quasi-transversality it is possible to devise magic-state distillation protocols.
We describe briefly a variant here, see also \cite{bravyi_universal_2005, bravyi_magic-state_2012, jones_multilevel_2013, campbell_unified_2017, campbell_unifying_2017, haah_codes_2018, hastings_distillation_2018}.

A magic state enables the implementation of some gate on another state.
The most common example is the state $\ket{A} = T\ket{+}$ which can be used to implement a $T$ gate using only a CNOT gate, a measurement and possibly a $S$ correction (see for example Figure~2 of \cite{campbell_roads_2017}).
If one uses a CSS code to encode information then the CNOT gate on the encoded level can be done transversally between two encoded blocks.
The main difficulty lies in obtaining an encoded magic state of good quality.

The usual starting point is one logical qubit being encoded in a \emph{base code} chosen to protect them from noise and with access to fault-tolerant implementations of all Clifford gates.
Then a common protocol consists in concatenating the base code with a \emph{distillation code}, say of parameters $\llbracket n,k,d\rrbracket $, which admits the implementation of logical $T$ gates on the encoded level by applying $T$ gates to the physical qubits.
Then using $n$ possibly low fidelity magic states encoded in the base code, one applies a transversal $T$ gate on $\ket{+}$ states at the level of the distillation code using the circuit described in the previous paragraph.
Measuring the checks of the distillation code conditioning on seeing a trivial syndrome and decoding to the base code one obtains $k$ magic states encoded in the base code of better quality.

Provided that the quality of the initial magic states is not too low, repeating the protocol sufficiently many times will reach any desired accuracy.
Then the amount of resources spent will directly depend on the parameters of the distillation code $\llbracket n,k,d\rrbracket $.
The efficiency of the protocol is often summarized in just one quantity:
\begin{equation}
\gamma = \frac{\log(n/k)}{\log(d)},\label{eq:gamma}
\end{equation}
since the average number of output distilled magic states at a desired accuracy, $\epsilon_{\rm out}$, per initial noisy magic state is given by $1/O\left (\log(\epsilon_{\rm out}^{-1}\right )^\gamma)$.
Hence, the smaller $\gamma$ the more efficient the distillation protocol.
Previously it was conjectured that $\gamma$ has to be at least $1$. 
However, it has recently been shown that $\gamma<1$ is achievable \cite{hastings_distillation_2018}.

\subsection{Puncturing techniques}
\label{sub:puncturingtechnique}

As explained above, most magic state distillation procedures rely on a distillation code which is a CSS code for which the $T$ gate is transversal and correspond to a transversal $T$ gate on the logical level of the code.
These codes are called tri-orthogonal codes \cite{bravyi_universal_2005, bravyi_magic-state_2012}, and rely on $3$-orthogonal or $3$-even spaces.

In more details and stated in general for the $R_\ell$ gate, this $\llbracket n,k,d\rrbracket$ distillation code, $\mathcal{C}_{\rm distill.}$, has to fulfill the conditions of exact transversality as per Prop.~\ref{prop:transv} as well as having 
\begin{equation}
	\left \vert\bs{x}L\right \vert = \sum_{j=1}^k x_j \pmod{2^\ell},\label{eq:polytransvT}
\end{equation}
where $L$ is a $k\times n$ generating matrix for the $X$-logical operators of $\mathcal{C}_{\rm distill.}$.
Fulfilling the exact transversality conditions of Prop.~\ref{prop:transv} ensures that one can transversally apply the $R_\ell$ gate on the code and \eqref{eq:polytransvT} ensures that the resulting operation on the logical level is the transversal $R_\ell$ gate.
If one has easy access to operations in the $(\ell-1)$th level of the Clifford hierarchy (as is usual in the case $\ell=3$) then the following weaker conditions are enough.
The code $\mathcal{C}_{\rm distill.}$ has to fulfill the conditions of quasi-transversality as per Prop.~\ref{prop:quasitransv} as well as having
\begin{equation}
	\left \vert\bs{x}L\right \vert = \sum_{j=1}^k x_j \pmod{2}.\label{eq:polyquasitransvT}
\end{equation}

Starting from an $\ell$-even, $\ell$-orthogonal space respectively, it is simple to obtain such codes by puncturing them \cite{haah_magic_2017, haah_codes_2018, hastings_distillation_2018}.
We describe this technique here and will apply it to some $\ell$-orthogonal spaces obtained from pin code relations in Section~\ref{sub:puncturingexamples}.
The idea goes as follows: consider an $\ell$-even subspace of $\mathbb{F}_2^n$ of dimension $m$ and take a $m\times n$ generating matrix, $G\in\mathbb{F}_2^{m\times n}$.
Using Gaussian elimination it is always possible to change the basis and put the matrix $G$ in the following form:
\begin{equation}
	G = \bordermatrix{
		~ & \leftarrow k \rightarrow & \leftarrow n \rightarrow \cr 
		k \updownarrow& \id & G_1\cr
		r \updownarrow& 0 & G_0},
\end{equation}
where $k$ is some chosen integer such that $1\leq k \leq m$ and $r=m-k$.
To obtain this form one just performs row operations.
This corresponds to puncturing the first $k$ positions of the space but note that one can perform some permutation of the columns before the Gaussian elimination which would yield some different $G_0$ and $G_1$ and corresponds to puncturing some other $k$ positions.
The rows of $G$ are a basis of the space and hence, by Prop.~\ref{prop:basisverif}, they verify \eqref{eq:multieven}.
We can deduce that the CSS code defined by $G_0$ as the generating matrix for the $X$-stabilizers and $G_1$ as the generating matrix for the $X$-logical operators verify the exact transversality conditions in Prop.~\ref{prop:transv}.
The $Z$-stabilizers of this code are given by the space dual to the space generated by the rows of both $G_0$ and $G_1$.
Moreover one can readily compute that 
\begin{equation}
	\left \vert\bs{x}L\right \vert = \left (2^\ell -1\right )\sum_{j=1}^k x_j \pmod{2^\ell}.\label{eq:puncttransvT}
\end{equation}
Note that we did not obtain a transversal $R_\ell$ on the logical level but rather $R_\ell^{-1}$.
The difference is $R_\ell^2$ which is in the $(\ell-1)$th level of the Clifford hierarchy.

Starting from a $\ell$-orthogonal space instead and going through the same procedure it is straightforward to see that one obtains a CSS code for which $R_\ell$ is quasi-transversal as per Prop.~\ref{prop:quasitransv} and that the logical operation is also the transversal $R_\ell$ up to a $(\ell-1)$th level correction.

These type of codes when used in a distillation procedure distill $n$ magic states into $k$ ones of better quality which depends on the distance of the code that has to be computed independently.

\section{Properties of quantum pin codes}
\label{sec:prop}

In this section we examine the properties of pin codes.
Since their definition is fairly general, their properties depend on the precise choice of pin code relations $F$.
We stay as general as possible and state precisely when the pin code relations need to be restricted.

\subsection{Code parameters and basic properties}
First we investigate the LDPC (Low Density Parity Check) property.
A code family is LDPC if it has stabilizer checks of constant weight and each of its qubits are acted upon by a constant number of checks.
LDPC codes were first considered by Gallager in the classical setting~\cite{gallager1962low}.
For pin codes, both properties depend on the relation $F$, but it is fairly easy to construct LDPC families.
For instance, pin codes based on Coxeter groups with fixed relations between generators and one growing compactifying relation can be LDPC, see Sec.~\ref{sec:hyperbolic}.
As another example, pin codes from chain complexes with fixed length $D+1$, sparse boundary map and growing dimension of the levels are LDPC as well.

Let us examine a simple example: choose some $D\in\mathbb{N}$, a set, $C$, of size $2m$ for some $m\in\mathbb{N}$ and the complete relation on $D+1$ copies of C: $F=C^{D+1}$.
One can easily verify that the relation $F$ is a pin code relation as $C$ has even cardinality.
The number of flags is $n_q = \vert F\vert = (2m)^{D+1}$ and the number of $x$- and $z$-pinned sets are $n_x=\binom{D+1}{x}\times(2m)^x$ and $n_z =\binom{D+1}{z}\times(2m)^z$.
If one considers growing $m$ then, the code would not be LDPC, but more strikingly the ratio of number of stabilizer checks to number of qubits would go to zero.
This illustrates that for a fixed $D$ the complete relation leads to very high rate and very low distance.
To get interesting codes, one either needs to vary $D$, or find some other relations with a number of flags growing significantly slower than the complete relation.

Concerning logical operators, we first note that they have even weight.
\begin{prop}[Logical operators have even weight]\label{prop:evenlogical}
	Let $F$ be a pin code relation on $D+1$ sets and let $\mathcal{C}$ be the associated $(x,z)$-pin code for $(x,z)\in\{1,\ldots,D\}^2$ with $x+z\leq D$.
	Then the $X$- and $Z$-logical operators of $\mathcal{C}$ have even weight. 
\end{prop}
\begin{IEEEproof}
	Let the set $L\subset F$ represent a $X$-logical operator, and let $T$ be a type of size $z$.
	Consider the set, S, of collections of pins given by the projection of type $T$ of the set $L$,
	\begin{equation}
	S = \Pi_T(L).
	\end{equation}
	For every $s\in S$, the pinned set $P_T(s)$ correspond to a $Z$-stabilizer and therefore has an even intersection with $L$.
	Pinned sets of the same type but defined by two different collections of pins are disjoint. 
	Hence, every element in $L$ appear in exactly one of the pinned sets $P_T(s)$ for some $s\in S$ and so the cardinal of $L$ is even.
	The proof for $Z$-logical operators is the same.	
\end{IEEEproof}
One can also prove the following general lower bound on the distance of pin codes.
\begin{prop}[Distance at least $2^{\min(x,z)+1}$]
	\label{prop:distance}
	Let $F$ be a pin code relation on $D+1$ sets and let $\mathcal{C}$ be the associated $(x,z)$-pin code for $(x,z)\in\{1,\ldots,D\}^2$ with $x+z\leq D$.
	Then the minimum distance of $\mathcal{C}$ is at least $2^{\min(x,z)+1}$. 
\end{prop}
\begin{IEEEproof}
	Without loss of generality we can assume $x\leq z$.
	Take any set $S_x \subset F$ containing less than $2^{x+1}$ distinct flags:
	\begin{equation}
		\left \vert S_x\right \vert < 2^{x+1}.
	\end{equation}
	
	Our goal is to  show that this cannot represent a logical Pauli operator.
	To do so we will exhibit a sequence of pins, each pin in this sequence will successively allow us to eliminate half or more of the remaining flags by pinning.
	All the pins of the sequence will define a pinned set corresponding to a stabilizer with overlap on exactly one flag with $S_x$ showing that $S_x$ cannot be a logical Pauli operator.
	
	Since all the flags in $S_x$ are distinct, there necessarily exists a rank $j_x$ such that $\left \vert \Pi_{\{j_x\}}(S_x) \right \vert > 1$, that is to say there exists a level on which they go through at least two different pins.
	We can use the pins in $\Pi_{\{j_x\}}(S_x)$ to partition $S_x$: 
	\begin{equation}
	S_x = \bigsqcup_{p\in \Pi_{\{j_x\}}(S_x)} \left (\mathcal{P}_{\{j_x\}}(p) \cap S_x\right ).
	\end{equation}
	Since this is a partition, one of these subsets necessarily contains half or less of $S_x$.
	That is to say there exists a pin $p_x\in L_{j_x}$ such that 
	\begin{equation}
	0<\left \vert P_{\{j_x\}}(p_x) \cap S_x \right \vert < 2^{x}.
	\end{equation}
	We define a new set, subset of $S_x$:
	\begin{equation}
	S_{x-1} = P_{\{j_x\}}(p_x) \cap S_x,
	\end{equation}
	which is such that
	\begin{equation}
	0<\left \vert S_{x-1} \right \vert < 2^{x}.
	\end{equation}
	There are then two cases: either $\left \vert S_{x-1}\right \vert = 1$ in which case we stop.
	Otherwise we iterate the process on $S_{x-1}$.
	This gives a sequence of ranks $\left (j_x,j_{x-1},\ldots\right )$ and pins of corresponding ranks $\left (p_x, p_{x-1},\ldots\right )$ and a sequence of sets of flags $S_x \supsetneq S_{x-1} \supsetneq \cdots$ which are such that
	\begin{align}
	S_{x-m} &= P_{\{j_{x-m+1}\}}(p_{x-m+1}) \cap S_{x-m+1},\\
	 0<&\left \vert S_{x-m} \right \vert < 2^{x-m+1}.\label{eq:sizeSx}
	 \end{align}
	Note that a rank chosen at some step cannot be chosen again at a subsequent step. Indeed by construction $\left \vert\Pi_{\{j_{x-k}\}}(S_{x-k-1}) \right \vert= \left \vert\{p_{x-k}\}\right \vert = 1$ and since for $m\geq k$ we have $S_{x-k-1} \supsetneq S_{x-m-1}$ then $\left \vert\Pi_{\{j_{x-k}\}}(S_{x-m-1}) \right \vert= \left \vert\{p_{x-k}\}\right \vert = 1$.
	From \eqref{eq:sizeSx} we can infer that we necessarily find a set with only one remaining flag, say $f$, i.e. $S_{x-k} = \{f\}$, and this in at most $x$ steps, hence we have $k \leq x$.
	For this set we have 
	\begin{align}
	\{f\} = S_{x-k} &= P_{\{j_{x-k+1}\}}(p_{x-k+1}) \cap S_{x-k+1}\nonumber\\
	&= P_{\{j_{x-k+1}\}}(p_{x-k+1}) \cap P_{\{j_{x-k+2}\}}(p_{x-k+2}) \cap S_{x-k+2}\nonumber\\
	&\vdots\nonumber\\
	&= \bigcap_{m=0}^{k-1} P_{\{j_{x-m}\}}(p_{x-m}) \cap S_{x}\nonumber\\
	\{f\}&= P_{\{j_x, j_{x-1},\ldots j_{x-k+1}\}}(p_x, p_{x-1}, \ldots, p_{x-k+1}) \cap S_x,
	\end{align}
	where we used the intersection property, Prop.~\ref{prop:interpinset}, for the last equality.
	Since $k\leq x\leq z$ and using the decomposition proposition Prop.~\ref{prop:pinsetdecomp}, the $k$-pinned set $P_{\{j_x, j_{x-1},\ldots j_{x-k+1}\}}(p_x, p_{x-1}, \ldots, p_{x-k+1})$ can represent both a $X$- and a $Z$-stabilizer and its intersection with $S_x$ is on only one flag, hence $S_x$ cannot represent a $X$- nor $Z$-logical operator.
	This concludes the proof that the distance of the code is at least $2^{x+1}$.
\end{IEEEproof}

In order to get odd weight logical operators, one has to introduce free pins, see Sec.~\ref{sec:bound}.
Note that in the presence of free pins, the proof above does not hold anymore.

Making precise statements about the dimension of pin codes is difficult in general.
To get closer to be able to do this we need to study the structure of the logical operators.

\subsection{Colored logical operators and unfolding}
\label{sec:unfold}
\begin{figure*}[h!]
	\centering
	\subfloat[]{\raisebox{.4in}{\includegraphics[width=.33\linewidth]{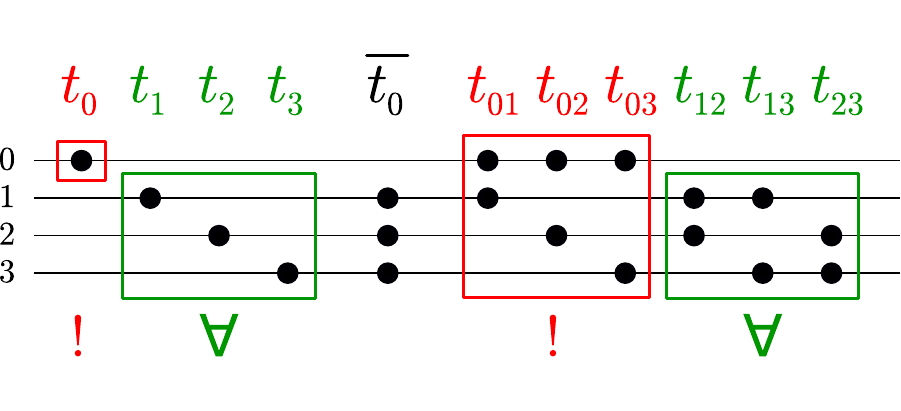}}%
		\label{fig:unfoldchart}}\hfil
	\subfloat[]{\includegraphics[height=1.8in]{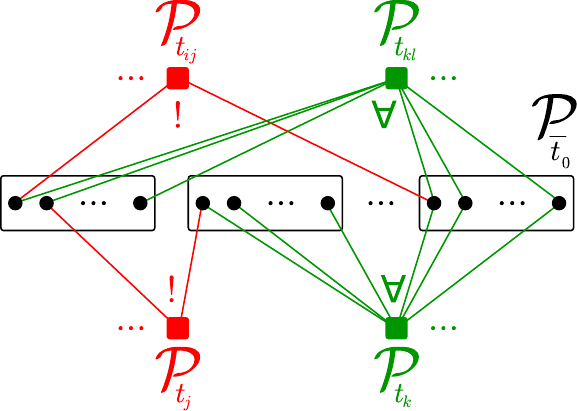}%
		\label{fig:unfold}}\hfil
	\subfloat[]{\includegraphics[height=1.8in]{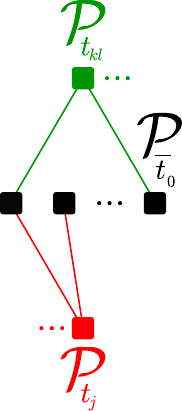}%
		\label{fig:unfoldreduced}}
	\caption{(a) Schematic representation of the different types for the $X$- and $Z$-stabilizers in the case of $D=3$.
		Each line represents a level.
		Types are represented by columns of dots where a dot is present if the level is in the type.
		The different types are classified according to their possible intersection with the complementary type to $t_0$.
		For some, the intersection, if not empty, is a unique flag, they are labeled by the symbol ``!''.
		For the others, the intersection, if not empty, is the full pinned set of type $\overline{t_0}$, they are labeled by the symbol ``$\forall$''.
		(b) The chain complex corresponding to the pin code, highlighting the intersections between the different sorts of $X$- and $Z$-stabilizers and pinned sets of type $\overline{t_0}$.
		(c) The $t_0$-shrunk chain complex derived from the pin code (see main text).}
	\label{fig:stringlogical}
\end{figure*}

\begin{figure*}[h]
	\centering
	\subfloat[]{\raisebox{.4in}{\includegraphics[width=.33\linewidth]{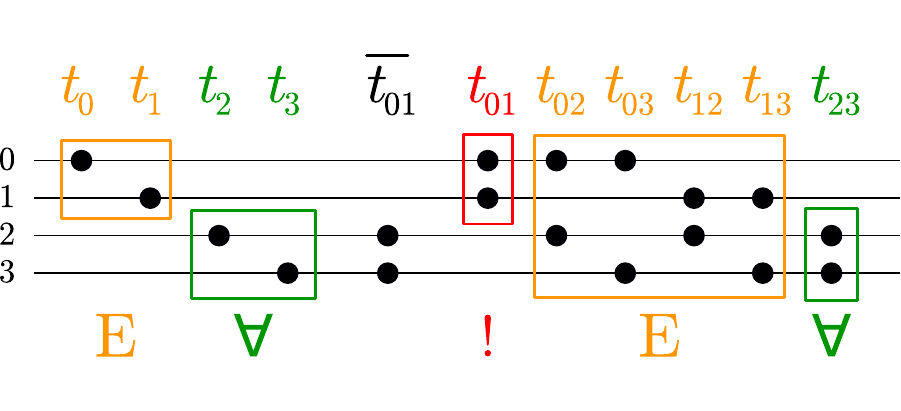}}%
		\label{fig:unfoldchart2}}\hfil
	\subfloat[]{\includegraphics[height=1.8in]{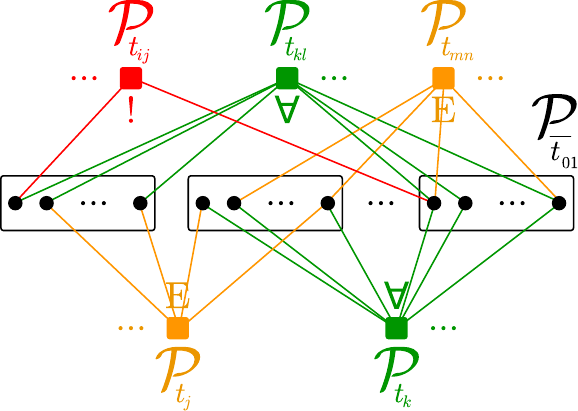}%
		\label{fig:unfold2}}\hfil
	\subfloat[]{\includegraphics[height=1.8in]{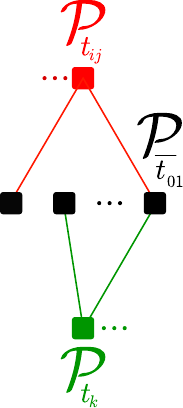}%
		\label{fig:unfoldreduced2}}
	\caption{(a) Schematic representation of the different types classified according to their possible intersection with the complementary type to $t_{01}$.
		Each line represents a level.
		Types are represented by columns of dots where a dot is present if the level is in the type.
		The ones with unique intersection are labeled with ``!'', the ones with even intersection with ``E'' and the one with full containment with ``$\forall$''.
		(b) The chain complex corresponding to the pin code, highlighting the intersections between the different sorts of $X$- and $Z$-stabilizers and pinned sets of type $\overline{t_{01}}$.
		(c) The $t_{01}$-shrunk (co)chain complex derived from the pin code (see main text).}
	\label{fig:membranelogical}
\end{figure*}

The structure of the logical operators of color codes is understood as colored string-nets or membrane nets \cite{bombin_exact_2007} and this structure is directly linked to an unfolding procedure existing for color codes \cite{kubica_unfolding_2015, bhagoji_equivalence_2015,kubica_efficient_2019}.
This structure mostly remains for all pin codes which means that we can map a given quantum pin code (seen as a chain complex) to several smaller chain complexes whose homological representatives can be lifted to logical operators of the original quantum pin code.
In the case of color codes these mappings fully characterize all logical operators.
While this is not true for a general quantum pin code it can bring insight into the structure of some of its logical operators. 

The general idea is to group qubits into sets with even overlap with all except one sort of stabilizer which will correspond to all stabilizers defined by pinned sets of a given type.
Logical operators build out of these sets then only depend on the structure of the one type of stabilizer selected.
Repeating this for different choices of type of stabilizer fully covers all logical operators in the case of color codes.

Consider a pin code relation, $F\subset L_0\times\cdots\times L_D$, and the associated $(x,z)$-pin code.
Define the complement of a type, $T$, denoted as $\overline{T}$:
\begin{equation}
\overline{T} = \{0,\ldots,D\}\setminus T.
\end{equation}
The intersection between a pinned set of type $T$ and a pinned set of type $\overline{T}$ is either empty or it contains exactly one flag.
Furthermore for any another type with the same number of pins as $T$, the corresponding pinned sets have necessarily even overlap with pinned sets of type $\overline{T}$, see Figs.~\ref{fig:unfoldchart} and \ref{fig:unfoldchart2} for visual representations of this.
This means that grouping flags according to pinned set of the complementary type $\overline{T}$ can single out logical operators only having to ensure commutation with pinned sets of type $T$.
For our code, $X$-stabilizers are generated by $x$-pin sets, which come in $\binom{D+1}{x}$ different types.
Take one such type, $T_x$, and group the qubits according to pinned sets of type $\overline{T_x}$.
Now the $Z$-stabilizers are generated by $z$-pinned sets, which come in $\binom{D+1}{z}$ different types.
Some of these types, we denote them as $T_z^{\rm inc.}$, are fully included in $\overline{T}$, which means that pinned sets of such type fully contain any group of qubits  they intersect.
The other types only partially intersect with the groups of qubits.
The situation is schematized in Fig.~\ref{fig:unfold} for $D=3$, $x=1$ and $z=2$.
From these considerations, one can construct a chain complex for which the homology gives candidate $Z$-logical operators.
Take the pinned sets of type $T_x$, for the level $0$, the pinned sets of type $\overline{T_x}$ for the level $1$, and the pinned sets of types $T_z^{\rm inc.}$ for level $2$ and the boundary map is given by the overlaps of these sets.
This is represented in Fig.~\ref{fig:unfoldreduced}, we call it the $T_x$-shrunk chain complex.
Then one can check that an element of the homology of this chain complex can be lifted to a potential $Z$-logical operator for the pin code.
Indeed it would commute with all the $X$-stabilizer, by homology for the stabilizers of type $T_x$ and by construction for the other $X$-stabilizers.
It would also not be simply generated by $Z$-stabilizers of type $T_z^{\rm inc.}$ by homology, and one would have to check for the other $Z$ types.
So it is a valid (potentially trivial) $Z$-logical operator.

The same procedure can be done for each of the $X$ types.
Symmetrically, the same can be done for the $X$-logical with the $Z$ types, and this is represented in Fig.~\ref{fig:membranelogical} in the  case $D=3$, $x=1$ and $z=2$.

Given a type $T$, the chain complexes constructed like this are called $T$-shrunk lattices in the case of color codes~\cite{bombin_exact_2007}.
For color codes obtained from the Wythoff construction described in Sec.~\ref{sec:construction_tilings}, the construction of the $T$-shrunk lattice is fairly direct.
First move the vertex from the middle of the fundamental simplex to the corner corresponding to the first rank in the type $T$, then focus on the opposite face: a simplex of dimension one less which now looks exactly like the beginning of the procedure but in a lower dimension.
Recursively exhaust all the ranks of $T$ in this way by each time adding a vertex in the middle of the current simplex and moving it to the corresponding corner.
 
These shrunk lattices are the basis for the unfolding procedure proved for color codes in all dimensions in \cite{kubica_unfolding_2015, kubica_efficient_2019}.
This procedure establishes a local unitary equivalence between a color codes and the reunion of the homological codes on the shrunk lattices corresponding to all the different types for $X$-stabilizers except one.
The local unitary acts separately on groups of qubits defined by the $X$-stabilizer generators of the type that is not used to produce one of the shrunk lattices.
The proof of the existence of the local unitary relies on the analysis of the so called overlap groups of stabilizers restricted to the support of the $X$-stabilizer generators aforementioned and the corresponding groups of qubits in the shrunk lattices.
The global structure is still present for general pin codes, but for the proof to hold we need to require that the linear dependency between the generators within the overlap groups in the pin code is such that the number of independent generators agrees with the number of independent generators in the corresponding shrunk lattices as an additional assumption.

These shrunk lattices are also the basis for some color code decoders \cite{delfosse_decoding_2014, aloshious_projecting_2018, kubica_three-dimensional_2018, kubica_efficient_2019} but these decoders rely on a lifting procedure from the shrunk lattices to the color code lattice which seems intrinsically geometric as it consists in finding a surface filling inside a boundary.
So it is at this point unclear how to leverage this structure in order to decode general pin codes.

\subsection{Gauge pin codes} 
\label{sec:gaugePC}
In this section we define gauge pin codes from a pin code relation.
Gauge pin codes can also be viewed as a generalization of gauge color codes \cite{bombin_gauge_2015}.

A gauge code, or subsystem code, is a code defined by a so called gauge group instead of a stabilizer group \cite{poulin_stabilizer_2005, bacon_operator_2006}.
For stabilizer codes, the code states are eponymously stabilized by the stabilizer group which is an abelian subgroup of the group of Pauli operators.
For gauge codes, the gauge group is not abelian and hence all gauge operators cannot share a common $+1$-eigenspace.
In this case the code states are stabilized by the center of the gauge group.
Gauge operators not in the center of the gauge group would qualify as logical operators in the case of a stabilizer code but are not used to encode information.

Take a pin code relation $F$ and two positive integers~$x$ and~$z$ such that $x+z<D$.
The associated pin code has its $X$-stabilizer generators defined by all the $x$-pinned sets and its $Z$-stabilizers generators defined by all the $z$-pinned sets.
Since the relation $F$ is a pin code relation, by Prop.~\ref{prop:evenoverlap} any $(D-x)$-pinned set has an even intersection with any $x$-pinned set.
So all the $(D-x)$-pinned sets correspond to some $Z$-logical operators.
On top of that, they generate all the $Z$-stabilizers.
Indeed using Prop.~\ref{prop:pinsetdecomp} and the fact that $D-x>z$ one shows that the $z$-pinned sets decompose into disjoint $(D-x)$-pinned sets.
As such $(D-x)$-pinned set define naturally $Z$-gauge operators which can be measured individually and whose outcomes can be recombined to reconstruct the value of the $Z$-stabilizers defined by $z$-pinned sets.
Symmetrically, the same happens for $(D-z)$-pinned sets which have even overlap with $z$-pinned sets and generate $x$-pinned sets and therefore can be viewed as $X$-gauge operators.

\begin{prop}[Gauge pin code]
	Given a pin code relation $F$ and two natural integers $x$ and $z$ such that $x+z<D$, one can define a gauge pin code by choosing $X$-gauge operators to be the $(D-z)$-pinned sets and $Z$-gauge operators to be the $(D-x)$-pinned sets.
	The corresponding stabilizer group for this gauge code contains the $x$-pinned sets as $X$-stabilizer generators and $z$-pinned sets as $Z$-stabilizer generators.
	The number of logical qubits is at most the number of logical qubits of the $(x, D-x)$-pin code or of the $(D-z, z)$-pin code obtained from the same relation $F$.
\end{prop}
\begin{IEEEproof}
	Let us first prove that $(D-z)$-pinned sets and $(D-x)$-pinned sets indeed define gauge operators that do not necessarily commute.
	For this, take any qubit in $F$, say $f \in F$, and any two types $T_x$ and $T_z$ of size $\vert T_x\vert = D-z$ and $\vert T_z\vert = D-x$, respectively. 
	Define the two collections $s_x = \Pi_{T_x}(f)$ and $s_z = \Pi_{T_z}(f)$, and the two $(D-z)$- and $(D-x)$-pinned sets, $P_{T_x}(s_x)$ and $P_{T_z}(s_z)$.
	One straightforwardly verifies that these pinned sets are constructed such that their intersection is a singleton, more precisely
	\begin{equation}
	P_{T_x}(s_x) \cap P_{T_z}(s_z) = \{f\},
	\end{equation}
	and they define a $X$-gauge operator and a $Z$-gauge operator, respectively, which therefore do not commute with one another.
	
	Let us prove now that the center of the gauge group, i.e. the stabilizer group, contains the $x$-pinned sets as $X$-stabilizers and $z$-pinned sets as $Z$-stabilizers.
	We have that $x<D-z$, hence by the decomposition property, Prop.~\ref{prop:pinsetdecomp}, $x$-pinned sets as $X$ operators are generated by $D-z$-pinned sets which are $X$ gauge operators.
	Similarly, $z<D-x$ and  $z$-pinned sets as $Z$ operators are generated by $D-x$-pinned sets which are $Z$ gauge operators.
	Moreover, 
	\begin{equation}
	x+(D-x)=z+(D-z)=D,
	\end{equation}
	and so by Proposition~\ref{prop:evenoverlap}, the $x$-pinned sets as $X$ operators and $z$-pinned sets as $Z$ operators commute with all gauge operators and belong to the center of the gauge group.
\end{IEEEproof}
Note that the stabilizer group can be larger than the group generated by $x$-pinned sets as $X$ stabilizers and $z$-pinned sets as $Z$ stabilizers. 
That is why in turn it can happen that the number of logical qubits is strictly smaller than that of the $(x, D-x)$-pin code or of the $(D-z, z)$-pin code from the same relation $F$.

The error correction procedure for a gauge code with only fully $X$-type or fully $Z$-type gauge operators is conveniently performed in two parts.
In one part, one measures the $X$-gauge operators, reconstructs the syndrome for the $X$-stabilizers and uses it to correct $Z$-errors.
In the other part, one measures the $Z$-gauge operators, reconstructs the syndrome for the $Z$-stabilizers and uses it to correct $X$-errors.

The advantages of this procedure in the case of gauge pin codes are two-fold.
First, the weight of the gauge generators, i.e. the number of qubits involved in each generator, is reduced compared to the weight of the stabilizer generators making their measurement easier and less error prone.
Second, the record of gauge operator measurements contains the information of the stabilizer measurements with redundancy.
To understand this redundancy consider a $x$-pinned set and define $k = (D-z)-x$.
This is the number of additional levels to pin in order to decompose the $x$-pinned set into $(D-z)$-pinned sets.
There are $\binom{D+1-x}{k}$ different ways to choose these additional levels to pin and therefore that many different ways to reconstruct the $x$-pinned set.
This redundancy permits a more robust syndrome extraction procedure which can even become in some cases single-shot, meaning that the syndrome measurements do not have to be repeated to reliably decode \cite{bombin_single-shot_2015}.
Meaning that even when the measurements are noisy one can measure the gauge operators only once and process the obtained information to reduce the noise enough and proceed with the computation.

\subsection{Transversality}

We examine here pin codes in regards of Prop.~\ref{prop:transv} and Prop.~\ref{prop:quasitransv}.
Nicely, $x$-pinned sets always have some multi-orthogonality property.
\begin{prop}[Multi-orthogonality of pinned sets]
	\label{prop:stabmultiorth}
	Let $F$ be a $(D+1)$-ary pin code relation. For any $x\in\{1,\ldots,D\}$, the $x$-pinned sets seen as binary vectors in $\mathbb{F}_2^F$ generate a $\lfloor D/x\rfloor$-orthogonal space.
\end{prop}
\begin{IEEEproof}
	Given $x\in\{1,\ldots,D\}$, we have that $x\cdot\lfloor D/x\rfloor \leq D$.
	Hence by Prop.~\ref{prop:interpinset}, the intersection of $\lfloor D/x\rfloor$ (or less) $x$-pinned sets is either empty or a pinned sets with at most $D$ pins.
	This pinned set with at most $D$ pins can always be decomposed by Prop.~\ref{prop:pinsetdecomp} into a disjoint union of $D$-pinned sets and hence has even weight as $F$ is a pin code relation as per Def.~\ref{def:pincoderel}.
\end{IEEEproof}
Interestingly it is also not too difficult to find pin code relations for which the $1$-pin sets are $D$-even.
For example, using a chain complex whose boundary map have even row and column weights and is regular enough will typically suffice.

One could also hope for the second part of Proposition~\ref{prop:quasitransv} to always holds.
Unfortunately it holds only partially in general.
\begin{prop}[$X$-logical intersection with $X$-stabilizers]
	Let $F$ be a $(D+1)$-ary pin code relation, and consider the associated $(x, z)$-pin code for $x\in\{1,\ldots,D\}$ and $z=D-x$.
	Then for any one $X$-logical operator, $\bs{L}$, and $k$ $X$-stabilizer generators, $\bs{G}^j$, with $k\leq \lfloor D/x\rfloor-1$,
	\[\left \vert\bs{L}\wedge\bs{G}^1\wedge\cdots\wedge\bs{G}^k\right \vert = 0 \pmod 2.\]
\end{prop}
\begin{IEEEproof}
	Indeed, using Prop.~\ref{prop:interpinset}, the overlap between $\lfloor D/x\rfloor - 1$ (or less) different $x$-pinned sets is either empty or a pinned set with at most $D-x = z$ pins. Hence by Prop.~\ref{prop:pinsetdecomp}, it can be decomposed into $z$-pinned sets, i.e. $Z$-stabilizers which have even overlap with $X$-logical operators by definition.
\end{IEEEproof}

Overlaps involving more than one $X$-logical operator do not have such guarantees in general.

Focusing on the case $\ell=3$, given the two propositions above the only problematic conditions are the ones of type
\begin{equation}
\left \vert \bs{L}^j\wedge\bs{L}^k\wedge \bs{G}^\ell\right \vert = 0 \pmod{2}.\label{eq:2LX}
\end{equation}
In order for these terms to hold, one has to have that the intersection of two $X$-logical operators is always a $Z$-logical operator.
This is the case for example for euclidean color codes.

In Sections~\ref{sub:puncturingexamples} and \ref{sec:ccz} we show how to use the multi-orthogonality properties of pinned-sets to get quantum codes with interesting transversal gates from any pin code relation. 
Another approach would be to restrict the pin code relations in order to devise subfamilies of quantum pin codes for which the condition \eqref{eq:2LX} is fulfilled.

\subsection{Boundaries and free pins}
\label{sec:bound}

The geometrical notion of colored boundaries existing for color codes can also be generalized to pin codes.
The way to do this is to introduce a specific type of pins which will be called \emph{free pins}.

Consider the chain complex approach to building pin code relations presented in Sec.~\ref{sec:chaincomplexeapproach}.
In this construction, it is sometimes necessary to add a rank-$0$ pin $b_0$ (in the level $L_0$) or a rank-$D$ pin $b_D$ (in the level $L_D$) in order to ensure that the relation $F$ is a pin code relation.
The new pin $b_0$ is linked to all the rank-$1$ pins which previously where linked to an odd number of rank-$0$ pins.
So even if the initial boundary relation is sparse, the number of connections to $b_0$ may be large.
As such the $1$-pinned set pinned by this new pin $b_0$ potentially contains a large number of flags.
To keep the size of the $1$-pinned sets under control it is then preferable to not allow to pin $b_0$ alone.
That is why we then call $b_0$ a free pin.
Any of the $D+1$ levels can contain free pins, the chain complex construction potentially put one in $L_0$ and one in $L_D$.
The rule for a larger collection of pins is that if it contains at least one non-free pin then it can define a valid pinned set, but if it is composed of only free pins then it is disregarded.
Finally consider when a flag is only composed of free pins, in that case this flag will not enter any valid pinned sets.
Hence such flags must also be discarded.
This is summarized in the following definition.

\begin{defi}[Pin code with free pins]
	Let $F$ be a pin code relation defined on $D+1$ levels of pins.
	Let some of the pins be labeled as free pins.
	Let $x$ and $z$ be two natural integers such that $x+z\leq D$.
	The associated $(x,z)$-pin code is defined as follows:
	The elements of $F$ containing at least one non-free pin are associated with qubits.
	All the $x$-pinned sets defined by a collection of pins containing at least one non-free pin are associated with $X$-stabilizer generators.
	All the $z$-pinned sets defined by a collection of pins containing at least one non-free pin are associated with $Z$-stabilizer generators.
\end{defi}

As examples we give a representation of Steane's $\llbracket 7,1,3\rrbracket $ code and the $\llbracket 4,2,2\rrbracket $ code as a $(1,1)$-pin codes with free pins in Figure~\ref{fig:steaneCC}.

\begin{figure}[h]
	\centering
	\includegraphics[height=.3\linewidth]{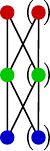}\qquad\qquad
	\includegraphics[height=.3\linewidth]{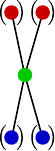}
	\caption{(Left) Representation of Steane's $\llbracket 7,1,3\rrbracket $ code as a $(1,1)$-pin code from a chain complex with free pins.
	There are three levels represented by the colors red green and blue.
	The free pins are represented between parenthesis.
	There are $8$ flags (all the paths going down from top to bottom) but one is composed only of free pins hence only $7$ qubits.
	There are three non-free pins defining three $1$-pinned sets for both $X$- and $Z$-stabilizers.
	(Right) Representation of the $\llbracket 4,2,2\rrbracket $ code as a $(1,1)$-pin code with free pins.
	There are four flags (all the paths going down from top to bottom), and a single non-free pin defining the $X$- and $Z$-stabilizer both containing the four flags.}
\label{fig:steaneCC}
\end{figure}

One idea to introduce free pins in every level could be to consider boundary map matrices which are almost sparse except for a small number of row or columns which could be dense.
The basis element corresponding to these would then be labeled as free pins in the construction of the pin code relation.

Note that in the presence of free pins, the proof of Prop.~\ref{prop:evenlogical} can only be reproduced when at least one level selected by the chosen type $T$ does not contain any free pin.
So as long as at least one level does not contain any free pin, it still holds that all logical operators have even weight.
When all levels contain at least one free pin then the code may contain odd weight logical operators.

The notion of free pins carries over straightforwardly to gauge pin codes.

\section{Examples and applications}
\label{sec:examples}
\subsection{Coxeter groups, hyperbolic color codes}
\label{sec:hyperbolic}
In Section~\ref{sec:construction_tilings} we discussed the construction of pin codes from tilings and Coxeter groups.
Well-known examples of such code families are color codes on euclidean tilings such as the hexagonal tiling in 2D and the bitruncated cubic honeycomb in 3D.
Using the Wythoff construction we can construct tilings which fulfill right pin code condition and therefore have the correct colorability for defining a color code.

\begin{figure}
	\centering
	\subfloat[]{{\includegraphics[width=.25\linewidth]{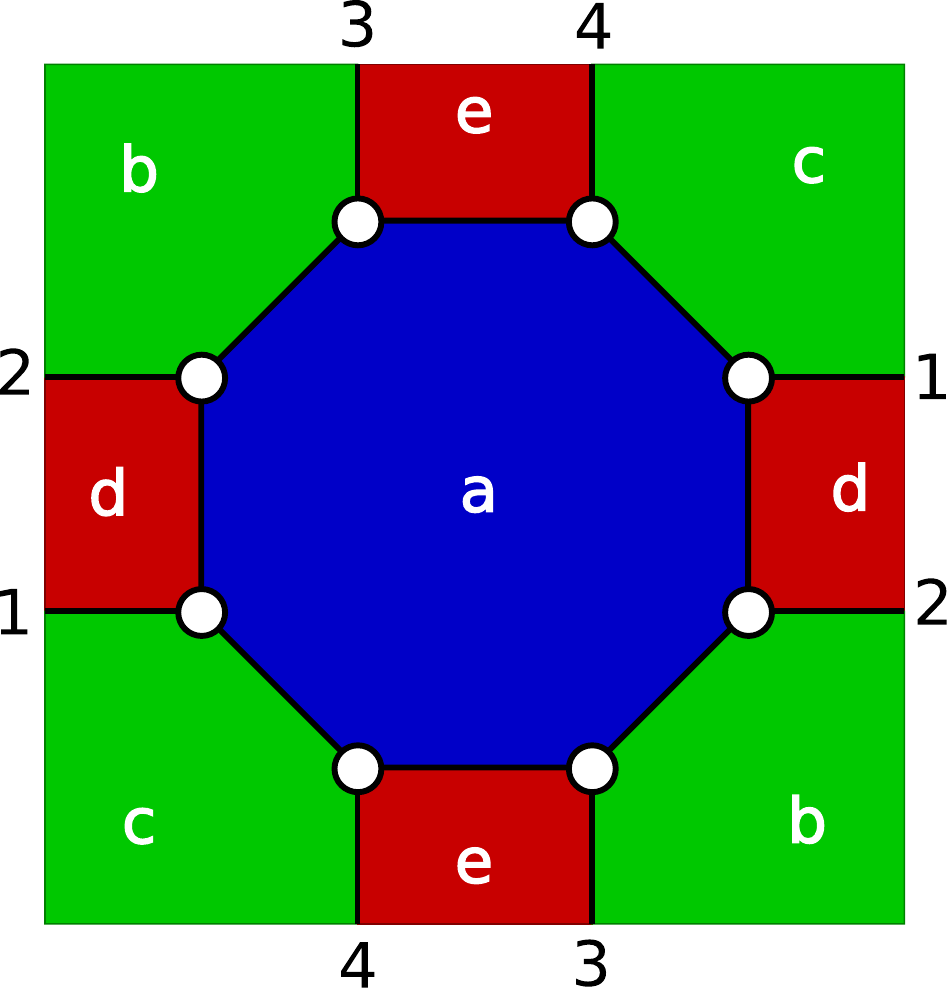}}%
		\label{fig:projective_plane_color_code}}\hfil
	\subfloat[]{{\includegraphics[width=.3\linewidth]{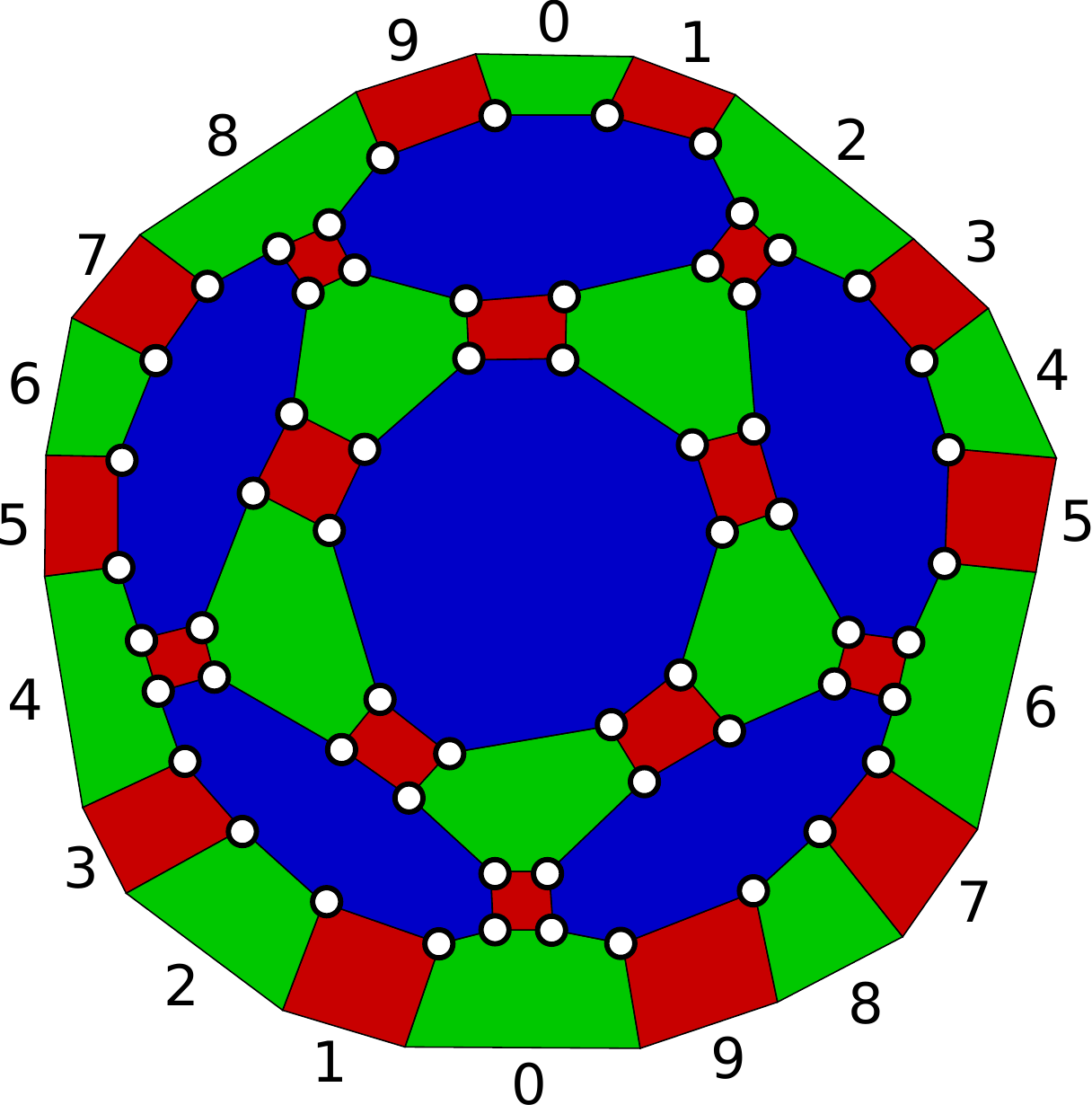}}%
		\label{fig:projective_plane_color_code_53}}\hfil
	\subfloat[]{\includegraphics[width=.25\linewidth]{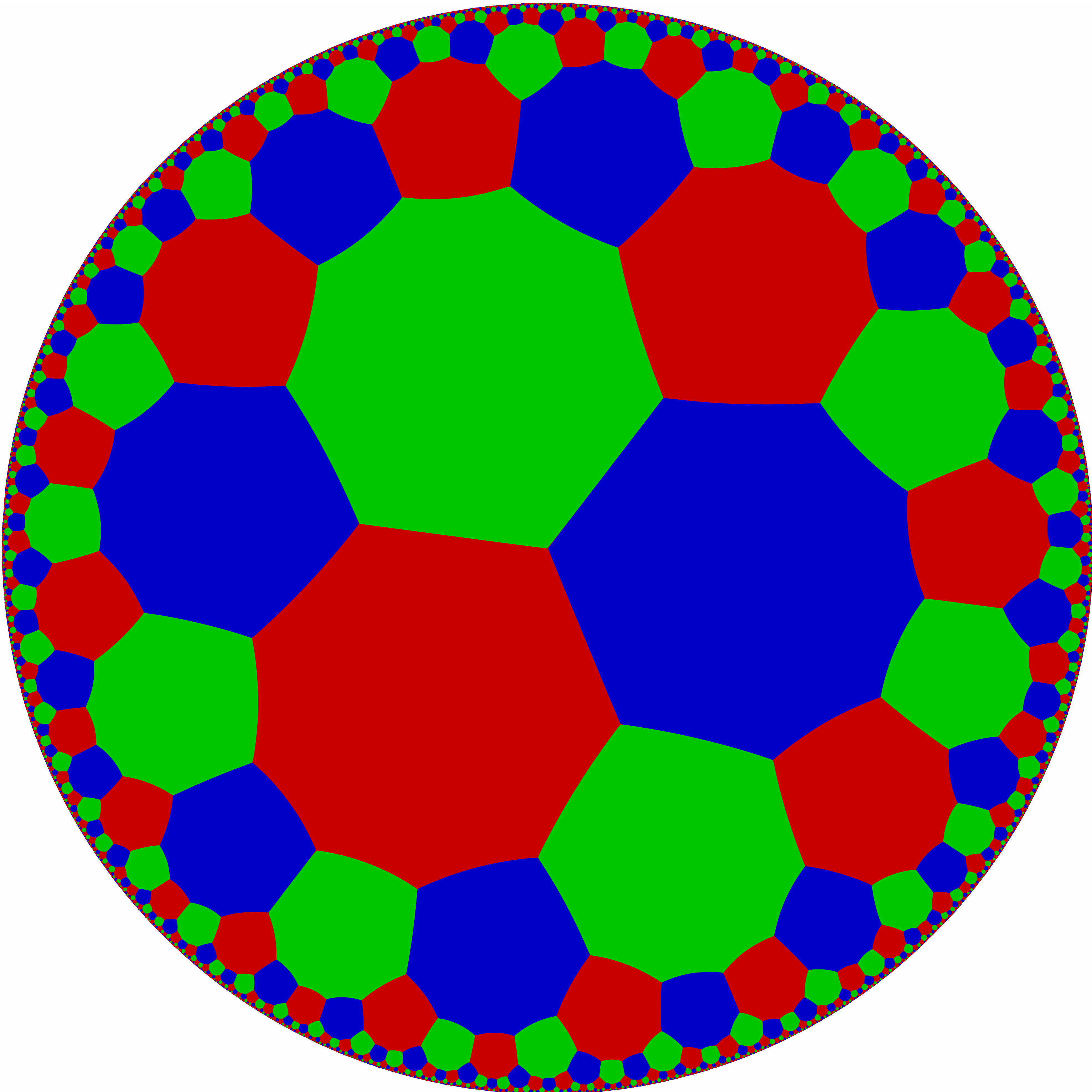}%
		\label{fig:hyperbolic_color_code}}
	\caption{(a)$\llbracket  8,2,2\rrbracket$   color code on the projective plane based on octahedral symmetry. The octahedron is topologically a sphere which can be turned into a projective plane by identifying opposite sides. 
		(b)~$\llbracket 60,2,6\rrbracket$   color code on the projective plane based dodecahedral symmetry.
		(c)~3-colored hyperbolic tiling with edges and vertices forming a 3-valent graph.
	}
	\label{fig:projective_plane}
\end{figure}

Besides the known euclidean examples we can consider tilings of more exotic spaces.
For the projective plane (cf.~\cite{freedman_projective_2001}) there exist two tilings based on the Wythoff construction: The first is based on the symmetry group of an octahedron.
It is an $\llbracket  8,2,2\rrbracket$  -code where the check generators correspond to one octagon, two red squares and two green squares, see Figure~\ref{fig:projective_plane_color_code}.
Note that this code does not quite fit the pin code definition because it contains distinct qubits which would be described by the same flag, for example $(d,c,a)$ on edge $1$.
This degeneracy explains why it escapes Prop.~\ref{prop:distance}.
The second is based on the icosahedral symmetry group, which gives a a $\llbracket  60,2,6\rrbracket$  -code with checks given by 6~decagons (blue), 10~hexagons (green) and 15~squares (red), see Figure~\ref{fig:projective_plane_color_code_53}.

Color codes based on two-dimensional hyperbolic tilings were first considered in \cite{delfosse_tradeoffs_2013} were 3-colorability and 3-valence was postulated (see Figure~\ref{fig:hyperbolic_color_code} for an example).
The Wythoff construction of  Section~\ref{sec:construction_tilings} allows us to obtain color codes from arbitrary regular tilings of closed hyperbolic surfaces.
To define a family of closed surfaces one needs to compactify the infinite lattice as explained in~\cite{breuckmann_constructions_2016}.
There are infinitely many regular tilings of 2D hyperbolic space. 
The lowest weight achievable with our construction is 4.8.10, meaning that checks are squares, octagons and dodecagons.
The smallest code in this family is $\llbracket  120,10,6\rrbracket$   based on a non-orientable hyperbolic surface (cf. Table 3.1 in \cite{breuckmann_homological_2017}).
Another small example is a $\llbracket  160,20,8\rrbracket$   code with stabilizer checks of weight~4 and~10 based on a 4.10.10 tiling of an orientable hyperbolic surface of genus~10.

Using the construction outlined in Section~\ref{sec:construction_tilings} we can consider any $D$-dimensional hyperbolic reflection group and obtain a tiling which is $D+1$-colorable and which has a $D+1$-valent graph.
In particular, we can consider hyperbolic tilings in 3D which are 4-colorable.
There exist four regular hyperbolic tilings in 3D of which two are self-dual tilings and two related by duality.
The self-dual ones are a tiling  by dodecahedra, denoted $\{5,3,5\}$, and one by icosahedra, denoted $\{3,5,3\}$.
The other are a tiling by cubes $\{4,3,5\}$ and its dual $\{5,3,4\}$.
All of these give rise to codes with maximum stabilizer weight 120.
Here we will focus on the $\{5,3,5\}$-tiling, which is the unique self-dual tiling of space by dodecahedra where five dodecahedra are placed around an edge.
Performing the Wythoff construction on a family of closed manifolds, all equipped with a $\{5,3,5\}$-tiling yields a code family where checks are of weight 20 and 120.
The weight of the stabilizer is given by the order of the subgroup of the full reflection group which is generated by all except for one of the generators.
The smallest example is a $\llbracket 7200,5526,4\rrbracket $ code.

\subsection{Pin codes from chain complexes}
\label{sec:hp}
In Sec.~\ref{sec:chaincomplexeapproach} we showed how from any $\mathbb{F}_2$ chain complex one can construct a pin code relation.
In this section we explore some specific examples of chain complexes and the corresponding pin codes.

One way to obtain arbitrary length chain complexes  is to use repeatedly the hypergraph product with a classical code.
The hypergraph product was introduced in \cite{tillich_quantum_2014} as a way to turn any two classical codes into a quantum code.
This product can be viewed as the tensor product of chain complexes, which takes two length-$2$ chain complexes to a length-$3$ chain complex.
More generally the product of a length-$k_1$ and length-$k_2$ chain complexes yields a length-$(k_1+k_2-1)$ chain complex.
This generalization and its characteristics has been studied in the context of homological codes \cite{audoux_tensor_2019, campbell_theory_2019, zengHigherDimensionalQuantumHypergraphProduct2019}.
We consider here the approach of \cite{zengHigherDimensionalQuantumHypergraphProduct2019} but look at the resulting chain complexes from the point of view of pin codes.

The idea goes as follows: consider $\mathcal{A}$, a $\mathbb{F}_2$ chain complex of length $k$, characterized by $\mathbb{F}_2$-vector spaces $(\mathcal{A}_j)_{0\leq j \leq k-1}$ and $(k-1)$ boundary maps $\partial^\mathcal{A}_j: \mathcal{A}_j \rightarrow \mathcal{A}_{j-1}$, obeying \eqref{eq:boundarymap}.
We now take the product with a chain complex of length $2$.
Note that any two vector spaces, $\mathcal{B}_1$ and $\mathcal{B}_0$ and any linear map between them $\partial^\mathcal{B}: \mathcal{B}_1\rightarrow\mathcal{B}_0$ defines a length-$2$ chain complex.
The product, $\mathcal{C}=\mathcal{A}\otimes\mathcal{B}$, is defined by $(k+1)$ vector spaces $\mathcal{C}_j$ for $0\leq j \leq k$,
\begin{equation}
\mathcal{C}_{j}=(\mathcal{B}_1\otimes\mathcal{A}_{j-1})\oplus(\mathcal{B}_0\otimes\mathcal{A}_{j}),
\end{equation}
with the convention that $\mathcal{A}_{-1}$ and $\mathcal{A}_k$ are both the zero vector spaces.
And the $k$ boundary maps, $\partial_j^\mathcal{C}:\mathcal{C}_j\rightarrow\mathcal{C}_{j-1}$, are defined as
\begin{align}
\forall u &= v\oplus w \in (\mathcal{B}_1\otimes\mathcal{A}_{j-1})\oplus(\mathcal{B}_0\otimes\mathcal{A}_{j})\nonumber\\
\partial_{j}^\mathcal{C}(u) &= (\id_{\mathcal{B}_1}\otimes \partial^\mathcal{A}_{j-1} + \partial^\mathcal{B}\otimes\id_{\mathcal{A}_{j-1}})(v) + (\id_{\mathcal{B}_0}\otimes \partial^\mathcal{A}_{j})(w).
\end{align}
One straightforwardly checks that the $\partial^\mathcal{C}_j$ are valid boundary maps, i.e. obeying \eqref{eq:boundarymap}.

Repeatedly taking the product with a length-$2$ chain complex therefore increase the length of the resulting chain complex each time by one.
Moreover any binary matrix defines a valid $\mathbb{F}_2$ chain complex of length $2$ so this approach allows to explore numerically many pin codes.

\begin{figure}[h]
	\centering
	\includegraphics[width=\linewidth]{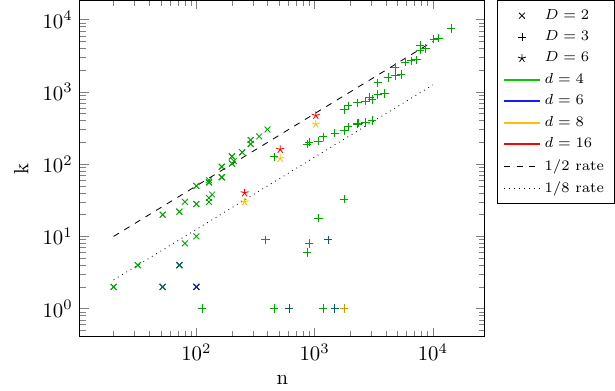}
	\caption{Plot of the $\llbracket n,k,d\rrbracket $ parameters of pin codes from chain complexes described in this section.
				One application of the hypergraph product on $3\times4$ binary matrices yields the $D=2$ pin codes represented by `$\times$' and two applications the $D=3$ pin codes represented by `$+$'.
				The $D=6$ pin codes described in Table.~\ref{tab:dim6pincodes} are represented by stars.
				The colors indicate the distance of the codes.
				The dashed and dotted lines represent the rates $k/n=1/2$ and $k/n=1/8$, respectively.}
	\label{fig:plotnk}
\end{figure}
\begin{figure}[h]
	\centering
	\includegraphics[width=\linewidth]{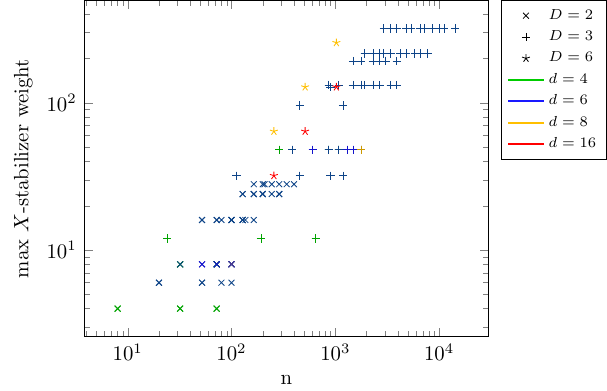}
	\caption{Plot of the maximum $X$-stabilizer weight of the pin codes from chain complexes described in this section.}
	\label{fig:plotnwx}
\end{figure}
We have looked at small binary matrices, up to $3\times 4$, and their self product to form pin code relations with $D=2$ and $D=3$.
We plot in Fig.~\ref{fig:plotnk} the code parameters obtained $\llbracket n,k,d\rrbracket $.
The distance is upper-bounded by numerically finding instances of low-weight logical operators.
Strinkingly these codes seem to show a general trend of high encoding rate $k/n$ for a small distance.
Indeed most of them are around $1/2$ rate but just distance $4$ which is the lower bound guaranteed by Prop.~\ref{prop:distance}.
A few of them seem to reach distance $6$ or $8$ but for significantly smaller rates.
The codes yielding no logical qubits are not displayed in this plot.
Note that this procedure is far from generating all chain complex of a given length.

As a comparison note that in \cite{grassl_codetables_2007} one can find codes with parameters such as $\llbracket104, 52, 10\rrbracket$, $\llbracket 104, 26, 15\rrbracket$ or $\llbracket 104, 13, 19\rrbracket$.
This is to illustrate that even at these high rates ($1/2$, $1/4$ and $1/8$) the distance can be higher.
Although these codes have larger stabilizer weights.

We have also looked at a few pin code relations for $D=6$ using small even size levels and the complete relation for~$F$.
Three notable examples are presented in Table.~\ref{tab:dim6pincodes}.
When writing $2^{\times6}\times4$ we mean that 6 of the levels contain each~$2$ pins and the last one contains $4$.
Since we use the complete relation, the number of flags and the size of the pinned sets are easily computed as a product of the size of some levels.
The number of logical qubits is computed numerically and for the distance we numerically found low-weight logical operators matching the lower bound of Prop.~\ref{prop:distance}.
All these examples give a good indication that the lower bound on the distance is tight.
\begin{table}[h]
	\centering
	\renewcommand{\arraystretch}{1.3}
	\caption{Parameters of some $D=6$ pin codes using the complete relation described by the size of the $D+1=7$ levels.}
	\label{tab:dim6pincodes}
	\begin{tabular}{c c c c }
		$(x,z)$ & $2^{\times6}\times4$ & $2^{\times5}\times4^{\times2}$ & $2^{\times4}\times4^{\times3}$ \\
		\hline
		$(2,4)$ &$\llbracket 256,30,8\rrbracket $& $\llbracket 512,120,8\rrbracket $ & $\llbracket 1024,358,8\rrbracket $\\
		$(3,3)$ &$\llbracket 256,40,16\rrbracket $& $\llbracket 512,160,16\rrbracket $ & $\llbracket 1024,472,16\rrbracket $
	\end{tabular}
\end{table}

We also represent the maximum weight of the $X$-stabilizers for these codes in Fig.~\ref{fig:plotnwx}.
When checking for transversal phase gates for $\ell=3$, most of the codes examined above do not satisfy \eqref{eq:2LX}.

\subsection{Puncturing triply-even spaces}
\label{sub:puncturingexamples}
\begin{figure}[h]
	\centering
	\includegraphics[width=.8\linewidth]{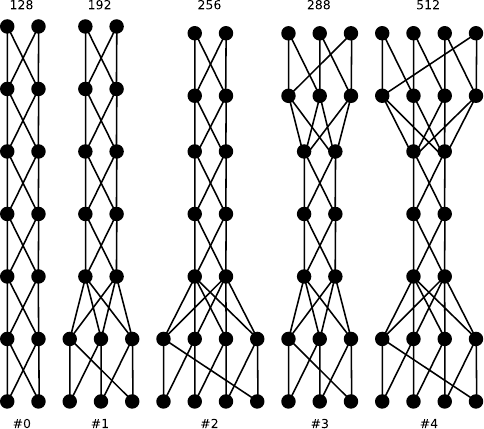}

	\caption{Some possible variations on Reed-Muller codes seen as pin codes on a chain complexes, by modifying the ends of the chain complex, the left most chain complex is used to define Reed-Muller codes.
	Each row of dots represent a level and the relation $F$ is given by all paths going down from top to bottom.
	Example of dimension $6$ are shown here, making variations on $\mathcal{RM}(2,7)$.
	The number of flags is written above each chain complex and an identification number is given below.}
	\label{fig:EiffelTCC}
\end{figure}

If pin codes in general are not guaranteed to fulfill all the requirements of Prop.~\ref{prop:transv} or Prop.~\ref{prop:quasitransv}, their stabilizers always form multi-orthogonal spaces, see Prop.~\ref{prop:stabmultiorth}.
This is directly useful as multi-orthogonal spaces together with puncturing techniques can be used to construct codes fulfilling Prop.~\ref{prop:quasitransv} (or Prop.~\ref{prop:transv} if the space is multi-even), see Section~\ref{sub:puncturingtechnique}.
We focus here on 3-even spaces (Def.~\ref{def:multieven}) and 3-orthogonal spaces (Def.~\ref{def:multiorth}) also referred to as triply-even spaces and tri-orthogonal spaces respectively.

In \cite{haah_codes_2018}, the authors use Reed-Muller codes, $\mathcal{RM}(r,m)$, to obtain initial tri-orthogonal spaces (also triply-even).
Viewed as generated by pinned-set, $\mathcal{RM}(r,m)$ comes from a simple chain complex, represented on the left of Figure~\ref{fig:EiffelTCC}, see also Section~\ref{sub:reedmuller}.
This chain complex can be modified in several ways to obtain different pin codes.
We tried different modifications in the case $D=6$, they are represented in Figure~\ref{fig:EiffelTCC}.
For all of them the $2$-pinned sets generate a triply-even space.

We tried to randomly puncture the pin codes obtained from these chain complexes; similarly to \cite{haah_codes_2018} but without deploying the more advanced techniques.
We were able to find a few interesting codes this way, see Table~\ref{tab:puncturedcodes}, which can be used to distill $T$ magic states.
The obtained parameters $\gamma$, see \eqref{eq:gamma}, are similar but do not improve on the small examples found in \cite{haah_codes_2018}.
\begin{table}[h]
	\centering
	\renewcommand{\arraystretch}{1.3}
	\caption{Some tri-orthogonal codes found by randomly puncturing the pin codes represented in Figure~\ref{fig:EiffelTCC}.}
	\label{tab:puncturedcodes}
	\begin{tabular}{c c c c}
		code \# & initial $n$ & punctured code: $\llbracket n,k,d\rrbracket $ & $\gamma = \frac{\ln\left (n/k\right )}{\ln d}$\\
		\hline
		$0$ &$128$& $\llbracket 116,12,4\rrbracket $ & $1.64$ \\
		$1$ &$192$& $\llbracket 175,17,4\rrbracket $ & $1.68$ \\
		$2$ &$256$& $\llbracket 236,20,4\rrbracket $ & $1.78$ \\
		$3$ &$288$& $\llbracket 261,27,4\rrbracket $ & $1.64$ \\
		$4$ &$512$& $\llbracket 466,46,4\rrbracket $ & $1.67$
	\end{tabular}
\end{table}

\subsection{Logical circuits of CCZs}
\label{sec:ccz}

It is also possible to use the property of multi-orthogonality of pinned sets on a given pin code relation in a slightly different way.
The construction proposed for Reed-Muller codes in \cite{rengaswamy_optimality_2020} can be directly adapted to general pin code relations.

\begin{prop}
	Let $\ell>1$ and $x>1$ be two integers and $F$ be a pin code relation on $x\ell$ sets ($D = x\ell-1$).
	The CSS code defined by $(x-1)$-pinned sets as $X$-stabilizer generators and $x$-pinned sets as $X$-logical operators satisfies Proposition~\ref{prop:quasitransv}.
\end{prop}
\begin{IEEEproof}
	We need for the condition $(i)$ to verify first that the $(x-1)$-pinned sets form a $\ell$-orthogonal space.
	According to Proposition~\ref{prop:stabmultiorth}, they form a $\left \lfloor \frac{D}{x-1}\right \rfloor$-orthogonal space and we can compute 
	\begin{equation}
		\left \lfloor \frac{D}{x-1}\right \rfloor = \left \lfloor\frac{x\ell-1}{x - 1}\right \rfloor = \left \lfloor\ell + \frac{\ell-1}{x-1}\right \rfloor\geq \ell.
	\end{equation}
	Hence $(i)$ is satisfied.
	For the condition $(ii)$ we use the intersection property (Prop.~\ref{prop:interpinset}):
	Taking two integers $s$ and $t$ such that $s\geq 1$, $t\geq 1$ and $s+t\leq\ell$, the intersection of $s$ $x$-pinned sets with $t$ $(x-1)$-pinned sets, gives a pinned sets with at most $(\ell -1)x+(x-1) = x\ell-1 = D$ pins.
	Since $F$ is a pin code relation $D$-pinned sets are always even.
	This proves $(ii)$.	
\end{IEEEproof}

Note that the parameters $D$ and $x$ where chosen not only to obtain a code guaranteed to satisfy Prop.~\ref{prop:quasitransv} but also where the logical operation realized belongs strictly to the level $\ell$ of the Clifford hierarchy because  some intersections of $\ell$ $X$-logical operators will be exactly one.
Indeed $X$-logical operators are $x$-pinned sets so intersecting $\ell$ of them can give (among other things) any pinned set with $x\ell = D+1$ pins which are either empty or singletons of one flag.

We can for example adapt the pin code relations presented in Fig.~\ref{fig:EiffelTCC} to have the correct dimension $D$ by inserting or removing levels of size $2$ in the middle of the chain complexes and look at what code parameters they give.
These parameters are compiled in Table.~\ref{tab:cczscodes}, for $D=5$ we remove the middle level and for $D=8$ we add two levels of size $2$ compared to Fig.~\ref{fig:EiffelTCC}.
All these codes support the transversal $T$ and up to a Clifford correction the logical operation implemented is some circuit of CC$Z$s characterized by which triple of $X$-logical operators have an odd overlap.

\begin{table}
	\centering
	\renewcommand{\arraystretch}{1.3}
	\caption{Alternative construction of CSS codes with transversal $T$ implementing some circuit of CC$Z$ gates on the logical level.}
	\label{tab:cczscodes}
		\begin{tabular}{c c}
		$D=5, x=2$ & $D=8, x=3$\\
		\hline
		$\llbracket 64,     15,   4\rrbracket $ & $\llbracket 512,   84,   8\rrbracket $\\
		$\llbracket 96,     23,   4\rrbracket $ & $\llbracket 768,   126, 8\rrbracket $\\
		$\llbracket 128,   31,   4\rrbracket $ & $\llbracket 1024, 168, 8\rrbracket $\\
		$\llbracket 144,   35,   4\rrbracket $ & $\llbracket 1152, 188, 8\rrbracket $\\
		$\llbracket 256,   63,   4\rrbracket $ & $\llbracket 2048, 332, 8\rrbracket $
	\end{tabular}
\end{table}

\section{Discussion}

Quantum pin codes form a large family of CSS codes which we have just begun to explore.
These codes can be viewed as a vast generalization of quantum color codes and the notions of boundaries, colored logical operators and shrunk lattices all generalize to pin codes.
Pin codes also have a gauge code version with potentially similar advantage as the gauge color codes.
The main property of pin codes is that their $X$- and $Z$-stabilizers form multi-orthogonal spaces.
We have presented two concrete ways of constructing pin codes and numerically explored some examples.
Several aspects of pin codes merit further studying.

First is finding restricted families with good parameters and LDPC property.
Exploring other finite groups with even order generators, other families of sparse chain complexes or finding other constructions of pin code relations altogether would help figuring out the achievable parameters for pin codes.

Second one concerns logical operators.
Understanding if some conditions on the pin code relation $F$ can make the logical operators fulfill the second condition of Prop.~\ref{prop:transv} or Prop.~\ref{prop:quasitransv} would help in the design of codes with transversal gates.
Also, logical operators and boundaries of 2D color codes have a richer structure than the colored logical operators and boundaries that we have explored, it would be interesting to generalize to pin codes with $D=2$ all the ones presented in~\cite{kesselring_boundaries_2018}, as well as for larger $D$.
Moreover, the structure of colored logical operators plays a key role in decoding color codes \cite{delfosse_decoding_2014, aloshious_projecting_2018, kubica_three-dimensional_2018, kubica_efficient_2019}.
Understanding if it can help in finding efficient decoders for more general pin codes is a natural question.

Finally more extensively exploring tri-orthogonal spaces obtained from pin code relations and puncturing them to obtain good $T$ distillation protocols as well as using them as the basis for $T$-to-CC$Z$ or other protocols seems promising, as distilling magic state will constitute a sizable fraction of any fault-tolerant quantum computation.

\appendices

\section{Proofs of Section~\ref{sec:multiorthogonality}}
\label{sec:proofmultiorthogonality}
This Appendix contains examples and proofs omitted in Section~\ref{sec:multiorthogonality}.

We first restate and prove Lemma~\ref{lem:idbinadd}:
\begin{replem}{lem:idbinadd}[Binary addition and integer addition]\label{lem:idbinaddbis}
	Denote binary addition with $\oplus$ and integer addition with $+$ or $\sum$.
	Given $r$ binary vectors $\bs{w}^1,\ldots,\bs{w}^r\in\mathbb{F}_2^n$ it holds that
	\begin{equation}
		\bigoplus_{m=1}^r\bs{w}^m = \sum_{s=1}^r(-2)^{s-1}\sum_{1\leq m_1<\cdots<m_s\leq r}\,\bigwedge_{i=1}^s\bs{w}^{m_i}.\label{eq:idbinaddbis}
	\end{equation}
	Similarly for their Hamming weights
	\begin{equation}
		\left \vert\bigoplus_{m=1}^r\bs{w}^m\right \vert = \sum_{s=1}^r(-2)^{s-1}\sum_{1\leq m_1<\cdots<m_s\leq r}\,\left \vert\bigwedge_{i=1}^s\bs{w}^{m_i}\right \vert.\label{eq:idbinaddhweightbis}
	\end{equation}
\end{replem}
It can be useful to work out an example for \eqref{eq:idbinaddbis} and \eqref{eq:idbinaddhweightbis}, here is one with $r=3$ and $n=4$
\begin{equation}
\begin{pmatrix}
0\\0\\0\\1
\end{pmatrix}\oplus
\begin{pmatrix}
0\\0\\1\\1
\end{pmatrix}\oplus
\begin{pmatrix}
0\\1\\1\\1
\end{pmatrix} = 
\begin{pmatrix}
0\\1\\0\\1
\end{pmatrix}, 
\end{equation}
which can also be computed in the following way
\begin{eqnarray}
\left [
\begin{pmatrix}
0\\0\\0\\1
\end{pmatrix} + \begin{pmatrix}
0\\0\\1\\1
\end{pmatrix} + \begin{pmatrix}
0\\1\\1\\1
\end{pmatrix} \right ] &\\- 2\left [
\begin{pmatrix}
0\\0\\0\\1
\end{pmatrix} + \begin{pmatrix}
0\\0\\0\\1
\end{pmatrix} + \begin{pmatrix}
0\\0\\1\\1
\end{pmatrix} \right ]& \\+ 4\begin{pmatrix}
0\\0\\0\\1
\end{pmatrix}&= 
\begin{pmatrix}
0\\1\\0\\1
\end{pmatrix} .
\end{eqnarray}
\begin{IEEEproof}
	Since the Hamming weight of a binary vector is simply the integer sum of its components we immediately have that $\eqref{eq:idbinaddbis}\Rightarrow\eqref{eq:idbinaddhweightbis}$.
	Moreover since the operations in \eqref{eq:idbinaddbis} are component-wise operations it is sufficient to prove it for $r$ binary numbers $w^1,\ldots,w^r\in\mathbb{F}_2$.
	One directly checks that for $r=2$
	\begin{equation}
		w^1\oplus w^2 = w^1 + w^2 - 2 w^1\wedge w^2.\label{eq:idbinaddinit}
	\end{equation}
	By induction suppose that \eqref{eq:idbinaddbis} holds for some $r\geq2$, then we have
	\begin{align}
		\bigoplus_{m=1}^{r+1}w^m =& \left (\bigoplus_{m=1}^rw^m\right )\oplus w^{r+1}\nonumber\\
		=& \left (\bigoplus_{m=1}^rw^m\right ) + w^{r+1} - 2\left (\bigoplus_{m=1}^rw^m\right )\wedge w^{r+1}\nonumber\\
		=& \sum_{s=1}^r(-2)^{s-1}\sum_{1\leq m_1<\cdots<m_s\leq r}\,\bigwedge_{i=1}^s\bs{w}^{m_i} + w^{r+1} \nonumber\\&- 2 \left (\sum_{s=1}^r(-2)^{s-1}\!\sum_{1\leq m_1<\cdots<m_s\leq r}\,\bigwedge_{i=1}^s\bs{w}^{m_i}\right )\wedge w^{r+1}\nonumber\\
		=&\sum_{s=1}^{r+1}(-2)^{s-1}\sum_{1\leq m_1<\cdots<m_s\leq r+1}\,\bigwedge_{i=1}^s\bs{w}^{m_i},
	\end{align}
	where we first used \eqref{eq:idbinaddinit} then the induction hypothesis then the distributivity of $\wedge$.
\end{IEEEproof}
We can now restate and prove Proposition~\ref{prop:equmultieven}:
\begin{repprop}{prop:equmultieven}[Characterization of multi-even spaces]\label{prop:equmultievenbis}
	Given $\ell\in\mathbb{N}$, a subspace $\mathcal{C}\subset\mathbb{F}_2^n$ is $\ell$-even if and only if for any integer $s\in\{1,\ldots,\ell\}$ and any $s$-tuple of vectors, $(\bs{v}^1,\ldots,\bs{v}^s)\in\mathcal{C}^s$, it holds that
	\begin{equation}
		\left|\bs{v}^1\wedge\cdots\wedge\bs{v}^s\right| = 0 \pmod{2^{\ell - s +1}}.\label{eq:multievenbis}
	\end{equation}
\end{repprop}
\begin{IEEEproof}
	The characterization given in Proposition~\ref{prop:equmultieven} implies the one given in Definition~\ref{def:multieven}, hence it suffice to prove the converse.
	The case $s=1$ is trivially verified, now suppose that for some $1\leq s < \ell$ one has for all $t\leq s$ and all $(\bs{v}^1,\ldots,\bs{v}^t)\in\mathcal{C}^t$
	\begin{equation}
		\left|\bs{v}^1\wedge\cdots\wedge\bs{v}^t\right| = 0 \pmod{2^{\ell - t +1}}.
		\label{eq:strongind}
	\end{equation}
	Take $s+1$ vectors $(\bs{v}^1,\ldots,\bs{v}^{s+1})\in\mathcal{C}^{s+1}$, using \eqref{eq:idbinaddhweightbis} from Lemma~\ref{lem:idbinadd} we can write
	\begin{align}
		(-2)^s\left|\bs{v}^1\wedge\cdots\wedge\bs{v}^{s+1}\right| =& \left \vert\bigoplus_{m=1}^{s+1}\bs{w}^m\right \vert \nonumber\\- \sum_{u=1}^{s}(-2)^{u-1}\!&\sum_{1\leq m_1<\cdots<m_u\leq s+1}\,\left \vert\bigwedge_{i=1}^u\bs{w}^{m_i}\right \vert\nonumber\\
		=& 0 - 0\pmod{2^\ell},
	\end{align}
	using \eqref{eq:strongind} to verify the divisibility by $2^\ell$ of the second part.
	Hence we have
	\begin{equation}
		\left|\bs{v}^1\wedge\cdots\wedge\bs{v}^{s+1}\right| = 0 \pmod{2^{\ell-s}}.
	\end{equation}
\end{IEEEproof}

Finally we restate and prove Proposition~\ref{prop:basisverif}:
\begin{repprop}{prop:basisverif}[multi-even/orthogonal space verification on a basis]\label{prop:basisverifbis}
	It is necessary and sufficient to verify \eqref{eq:multieven} or \eqref{eq:multiorth} on a basis to ensure that a space is multi-even or multi-orthogonal respectively.
\end{repprop}
\begin{IEEEproof}
	As basis vectors belong to the space it is immediate to check that they necessarily verify \eqref{eq:multieven} or \eqref{eq:multiorth} respectively, if the space is multi-even or multi-orthogonal respectively.
	The sufficient proof for both the multi-even and multi-orthogonal cases relies on decomposing vectors over the basis $\mathcal{B}=\left \{\bs{b}^1,\ldots,\bs{b}^m\right \}$ which verifies the property (either \eqref{eq:multieven} or \eqref{eq:multiorth}) and using \eqref{eq:idbinaddhweight} several times.
	Pick $s\in\{1,\ldots,\ell\}$ and $s$ vectors $\left (\bs{v}^1,\ldots,\bs{v}^s\right )\in\mathcal{C}^s$.
	Each of the vectors $\bs{v}^j$ can be decomposed over $\mathcal{B}$:
	\begin{equation}
		\bs{v}^j = \bigoplus_{i=1}^{p_j}\bs{b}^{k^j_i},
	\end{equation}
	where $p_j$ is the number of basis elements in the decomposition of $\bs{v}^j$ over $\mathcal{B}$ and the $k^j_i$ their indices.
	We can now rewrite the Hamming weight of the element-wise product by using the decomposition of the $\bs{v}^j$ and \eqref{eq:idbinaddhweight} one after the other.
	This is done in \eqref{eq:bigwedgesum}.
	\begin{figure*}[!t]
		\normalsize
		\setcounter{MYtempeqncnt}{\value{equation}}
		\setcounter{equation}{77}
	\begin{align}
		\left \vert \bs{v}^{1}\wedge\bs{v}^2\wedge\cdots\wedge\bs{v}^s\right \vert =& \left \vert \bigoplus_{i=1}^{p_1}\bs{b}^{k^1_i}\wedge\bs{v}^2\cdots\wedge\bs{v}^{s}\right \vert\nonumber\\
		=&\sum_{q_1=1}^{p_1}(-2)^{q_1-1}\sum_{1\leq m^1_1<\cdots<m^1_{q_1}\leq p_1}\left \vert \bigwedge_{r=1}^{q_1}\bs{b}^{k^1_{m^1_r}}\wedge\bs{v}^2\cdots\wedge\bs{v}^{s}\right \vert\nonumber\\
		=&\sum_{q_1,\ldots ,q_s=1}^{p_1,\ldots,p_s}(-2)^{\sum_{j=1}^s(q_j-1)}\sum_{\substack{1\leq m^1_1<\cdots<m^1_{q_1}\leq p_1\\\cdots\\1\leq m^s_1<\cdots<m^s_{q_s}\leq p_s}}\left \vert \bigwedge_{r=1}^{q_1}\bs{b}^{k^1_{m^1_r}}\wedge\cdots\bigwedge_{r=1}^{q_s}\bs{b}^{k^s_{m^s_r}}\right \vert.\label{eq:bigwedgesum}
	\end{align}
	\setcounter{equation}{\value{MYtempeqncnt}}
	\hrulefill
	\vspace*{4pt}
\end{figure*}
	
	Considering expression \eqref{eq:bigwedgesum}, the case of the multi-orthogonal property is the most straightforward.
	In this case we assume that \eqref{eq:multiorth} holds for the elements of $\mathcal{B}$ and choose $s=\ell$.
	Elements of the sum in \eqref{eq:bigwedgesum} with $\ell$ basis vector in the element-wise product will be even by \eqref{eq:multiorth} and the other elements have a prefactor of two.	
	This proves that $\left \vert \bs{v}^1\wedge\cdots\wedge\bs{v}^\ell\right \vert=0\pmod 2$ and that the space is $\ell$-orthogonal.
	
	For the case of the multi-even property, we assume that \eqref{eq:multieven} holds for the elements of $\mathcal{B}$.
	Elements of the sum in \eqref{eq:bigwedgesum} where $\sum q_j \leq \ell$ have a Hamming weight for the element-wise product which is zero modulo $2^{\ell - \sum q_j +1}$ by \eqref{eq:multieven}.
	Moreover they are multiplied by $2^{\sum q_j - s}$ and hence are zero modulo $2^{\ell - s +1}$.
	The other terms have a prefactor of $2^{\ell -s +1}$.
	Hence $\left \vert \bs{v}^1\wedge\cdots\wedge\bs{v}^s\right \vert=0\pmod {2^{\ell-s+1}}$ and the space is $\ell$-even.
	
\end{IEEEproof}

\section{Quasi-transversality}
\label{sec:multiorth}

In this appendix we detail the three weighted polynomial in \eqref{eq:Flx}, \eqref{eq:Fly} and {eq:Flxy} which determine the transversal action of $R_\ell$ on code states.
Using identity \eqref{eq:idbinadd} and denoting $\bs{L}^m$ as the $m^\text{th}$ row of matrix $L$ and $\bs{G}^n$ as the $n^\text{th}$ row of matrix $G$ we can write 
\begin{align}
F_\ell(\bs{x}) &= \left \vert\bs{x}L\right \vert\nonumber\\
&= \sum_{s=1}^\ell(-2)^{s-1}\sum_{1\leq m_i\leq k}\,\left\vert\bigwedge_{i=1}^s \bs{L}^{m_i}\right \vert\prod_{i=1}^sx_{m_i},\\[1em]
F^\prime_\ell(\bs{y}) &= \left \vert\bs{y}G\right \vert \nonumber\\
&= \sum_{t=1}^\ell(-2)^{t-1}\sum_{1\leq n_j\leq r}\,\left\vert\bigwedge_{j=1}^t \bs{G}^{n_j}\right \vert\prod_{j=1}^ty_{n_j},\\[1em]
F^{\prime\prime}_\ell(\bs{x}, \bs{y}) &= -2\left \vert\bs{x}L\wedge\bs{y}G\right \vert = \nonumber\\
\sum_{\substack{s+t=2\\s\geq1,\,t\geq1}}^{\ell}(-2&)^{s+t-1}\sum_{\substack{1\leq m_i\leq k\\1\leq n_j\leq r}}\,\left\vert\bigwedge_{i=1}^s \bs{L}^{m_i}\bigwedge_{j=1}^t\bs{G}^{n_j}\right \vert\prod_{i=1}^sx_{m_i}\prod_{j=1}^{t} y_{n_j}.
\end{align}

One can readily see that these are all correctly weighted polynomial, i.e. with a prefactor of $2^{s-1}$ in front of monomials of degree $s$, and their coefficients are given by the size of the overlaps between rows of the matrices $L$ or $G$. 

We can check that \eqref{eq:l-1correction} follows from Proposition~\ref{prop:quasitransv}.
Indeed, assuming Prop.~\ref{prop:quasitransv} holds, then $(i)$ enforces that all coefficients $\left\vert\bigwedge_{j=1}^t \bs{G}^{n_j}\right \vert$ are divisible by $2$ and $(ii)$ that all coefficients $\left\vert\bigwedge_{i=1}^s \bs{L}^{m_i}\bigwedge_{j=1}^t\bs{G}^{n_j}\right \vert$ also are divisible by $2$.
Hence we can pull out a factor $2$ in front of everything while keeping the correct prefactor in front of each monomial.

\section*{Acknowledgment}
C.V. would like to thank B. Audoux, E.T. Campbell and L.P. Pryadko for fruitful discussions at different stages of this project.
C.V. acknowledge support by the European Research Council (EQEC, ERC Consolidator Grant No: 682726) as well as a QuantERA grant for the QCDA consortium. NPB is supported by the UCLQ fellowship.

\ifCLASSOPTIONcaptionsoff
  \newpage
\fi

\bibliographystyle{IEEEtran}

\bibliography{IEEEabrv,PinCodes_final}

\begin{thebibliography}{10}
\providecommand{\url}[1]{#1}
\csname url@samestyle\endcsname
\providecommand{\newblock}{\relax}
\providecommand{\bibinfo}[2]{#2}
\providecommand{\BIBentrySTDinterwordspacing}{\spaceskip=0pt\relax}
\providecommand{\BIBentryALTinterwordstretchfactor}{4}
\providecommand{\BIBentryALTinterwordspacing}{\spaceskip=\fontdimen2\font plus
\BIBentryALTinterwordstretchfactor\fontdimen3\font minus
  \fontdimen4\font\relax}
\providecommand{\BIBforeignlanguage}[2]{{%
\expandafter\ifx\csname l@#1\endcsname\relax
\typeout{** WARNING: IEEEtran.bst: No hyphenation pattern has been}%
\typeout{** loaded for the language `#1'. Using the pattern for}%
\typeout{** the default language instead.}%
\else
\language=\csname l@#1\endcsname
\fi
#2}}
\providecommand{\BIBdecl}{\relax}
\BIBdecl

\bibitem{campbell_roads_2017}
\BIBentryALTinterwordspacing
E.~T. Campbell, B.~M. Terhal, and C.~Vuillot, ``\BIBforeignlanguage{en}{Roads
  towards fault-tolerant universal quantum computation},''
  \emph{\BIBforeignlanguage{en}{Nature}}, vol. 549, no. 7671, pp. 172--179,
  Sep. 2017. [Online]. Available:
  \url{https://www.nature.com/articles/nature23460}
\BIBentrySTDinterwordspacing

\bibitem{calderbank_good_1996}
\BIBentryALTinterwordspacing
A.~R. Calderbank and P.~W. Shor, ``Good quantum error-correcting codes exist,''
  \emph{Physical Review A}, vol.~54, no.~2, pp. 1098--1105, Aug. 1996.
  [Online]. Available: \url{https://link.aps.org/doi/10.1103/PhysRevA.54.1098}
\BIBentrySTDinterwordspacing

\bibitem{steane_andrew_multiple-particle_1996}
\BIBentryALTinterwordspacing
{Steane Andrew}, ``Multiple-particle interference and quantum error
  correction,'' \emph{Proceedings of the Royal Society of London. Series A:
  Mathematical, Physical and Engineering Sciences}, vol. 452, no. 1954, pp.
  2551--2577, Nov. 1996. [Online]. Available:
  \url{https://royalsocietypublishing.org/doi/10.1098/rspa.1996.0136}
\BIBentrySTDinterwordspacing

\bibitem{gottesman_stabilizer_1997}
D.~Gottesman, ``Stabilizer codes and quantum error correction,'' {PhD}
  {Thesis}, California Institute of Technology, Jan. 1997.

\bibitem{calderbank_quantum_1997}
A.~R. Calderbank, E.~M. Rains, P.~W. Shor, and N.~J.~A. Sloane, ``Quantum error
  correction via codes over {GF}(4),'' in \emph{Proceedings of {IEEE}
  {International} {Symposium} on {Information} {Theory}}, Jun. 1997, pp. 292--.

\bibitem{grassl_codes_1997}
\BIBentryALTinterwordspacing
M.~Grassl, T.~Beth, and T.~Pellizzari, ``Codes for the quantum erasure
  channel,'' \emph{Physical Review A}, vol.~56, no.~1, pp. 33--38, Jul. 1997.
  [Online]. Available: \url{https://link.aps.org/doi/10.1103/PhysRevA.56.33}
\BIBentrySTDinterwordspacing

\bibitem{freedman_projective_2001}
\BIBentryALTinterwordspacing
M.~H. Freedman and D.~A. Meyer, ``\BIBforeignlanguage{en}{Projective {Plane}
  and {Planar} {Quantum} {Codes}},'' \emph{\BIBforeignlanguage{en}{Foundations
  of Computational Mathematics}}, vol.~1, no.~3, pp. 325--332, Jul. 2001.
  [Online]. Available: \url{https://doi.org/10.1007/s102080010013}
\BIBentrySTDinterwordspacing

\bibitem{freedman_$z_2$-systolic_2002}
M.~H. Freedman, D.~A. Meyer, and F.~Luo, ``\${Z}\_2\$-systolic freedom and
  quantum codes,'' in \emph{Mathematics of quantum computation}, ser.
  Computational {Mathematics}, R.~K. Brylinski and G.~Chen, Eds.\hskip 1em plus
  0.5em minus 0.4em\relax Boca Raton, FL: Chapman \& Hall/CRC, 2002, no.~3, pp.
  287--320.

\bibitem{kitaev_fault-tolerant_2003}
\BIBentryALTinterwordspacing
A.~Y. Kitaev, ``Fault-tolerant quantum computation by anyons,'' \emph{Annals of
  Physics}, vol. 303, no.~1, pp. 2--30, Jan. 2003. [Online]. Available:
  \url{http://www.sciencedirect.com/science/article/pii/S0003491602000180}
\BIBentrySTDinterwordspacing

\bibitem{bacon_operator_2006}
\BIBentryALTinterwordspacing
D.~Bacon, ``Operator quantum error-correcting subsystems for self-correcting
  quantum memories,'' \emph{Physical Review A}, vol.~73, no.~1, p. 012340, Jan.
  2006. [Online]. Available:
  \url{https://link.aps.org/doi/10.1103/PhysRevA.73.012340}
\BIBentrySTDinterwordspacing

\bibitem{bombin_topological_2006}
\BIBentryALTinterwordspacing
H.~Bombin and M.~A. Martin-Delgado, ``Topological {Quantum} {Distillation},''
  \emph{Physical Review Letters}, vol.~97, no.~18, p. 180501, Oct. 2006.
  [Online]. Available:
  \url{https://link.aps.org/doi/10.1103/PhysRevLett.97.180501}
\BIBentrySTDinterwordspacing

\bibitem{kovalev_quantum_2013}
\BIBentryALTinterwordspacing
A.~A. Kovalev and L.~P. Pryadko, ``Quantum {Kronecker} sum-product low-density
  parity-check codes with finite rate,'' \emph{Physical Review A}, vol.~88,
  no.~1, p. 012311, Jul. 2013. [Online]. Available:
  \url{https://link.aps.org/doi/10.1103/PhysRevA.88.012311}
\BIBentrySTDinterwordspacing

\bibitem{couvreur_construction_2013}
A.~Couvreur, N.~Delfosse, and G.~Zémor, ``A {Construction} of {Quantum} {LDPC}
  {Codes} {From} {Cayley} {Graphs},'' \emph{IEEE Transactions on Information
  Theory}, vol.~59, no.~9, pp. 6087--6098, Sep. 2013.

\bibitem{guth_quantum_2014}
\BIBentryALTinterwordspacing
L.~Guth and A.~Lubotzky, ``Quantum error correcting codes and 4-dimensional
  arithmetic hyperbolic manifolds,'' \emph{Journal of Mathematical Physics},
  vol.~55, no.~8, p. 082202, Aug. 2014. [Online]. Available:
  \url{https://aip.scitation.org/doi/10.1063/1.4891487}
\BIBentrySTDinterwordspacing

\bibitem{tillich_quantum_2014}
J.~Tillich and G.~Zémor, ``Quantum {LDPC} {Codes} {With} {Positive} {Rate} and
  {Minimum} {Distance} {Proportional} to the {Square} {Root} of the
  {Blocklength},'' \emph{IEEE Transactions on Information Theory}, vol.~60,
  no.~2, pp. 1193--1202, Feb. 2014.

\bibitem{audoux_tensor_2019}
\BIBentryALTinterwordspacing
B.~Audoux and A.~Couvreur, ``On tensor products of {CSS} codes,'' \emph{Annales
  de l’Institut Henri Poincaré D}, vol.~6, no.~2, pp. 239--287, Mar. 2019.
  [Online]. Available:
  \url{https://www.ems-ph.org/journals/show_abstract.php?issn=2308-5827&vol=6&iss=2&rank=4}
\BIBentrySTDinterwordspacing

\bibitem{leverrier_quantum_2015}
A.~Leverrier, J.~Tillich, and G.~Zémor, ``Quantum {Expander} {Codes},'' in
  \emph{2015 {IEEE} 56th {Annual} {Symposium} on {Foundations} of {Computer}
  {Science}}, Oct. 2015, pp. 810--824.

\bibitem{bravyi_universal_2005}
\BIBentryALTinterwordspacing
S.~Bravyi and A.~Kitaev, ``Universal quantum computation with ideal {Clifford}
  gates and noisy ancillas,'' \emph{Physical Review A}, vol.~71, no.~2, p.
  022316, Feb. 2005. [Online]. Available:
  \url{https://link.aps.org/doi/10.1103/PhysRevA.71.022316}
\BIBentrySTDinterwordspacing

\bibitem{bombin_gauge_2015}
\BIBentryALTinterwordspacing
H.~Bombín, ``\BIBforeignlanguage{en}{Gauge color codes: optimal transversal
  gates and gauge fixing in topological stabilizer codes},''
  \emph{\BIBforeignlanguage{en}{New Journal of Physics}}, vol.~17, no.~8, p.
  083002, Aug. 2015. [Online]. Available:
  \url{https://doi.org/10.1088%2F1367-2630%2F17%2F8%2F083002}
\BIBentrySTDinterwordspacing

\bibitem{bombin_exact_2007}
\BIBentryALTinterwordspacing
H.~Bombin and M.~A. Martin-Delgado, ``Exact topological quantum order in
  \{\${d}=3\$\} and beyond: {Branyons} and brane-net condensates,''
  \emph{Physical Review B}, vol.~75, no.~7, p. 075103, Feb. 2007. [Online].
  Available: \url{https://link.aps.org/doi/10.1103/PhysRevB.75.075103}
\BIBentrySTDinterwordspacing

\bibitem{steaneQuantumReedMullerCodes1999}
A.~Steane, ``Quantum {{Reed-Muller}} codes,'' \emph{IEEE Transactions on
  Information Theory}, vol.~45, no.~5, pp. 1701--1703, Jul. 1999.

\bibitem{zhangQuantumReedMullerCodes1997}
L.~Zhang and I.~Fuss, ``Quantum {{Reed-Muller Codes}},''
  \emph{arXiv:quant-ph/9703045}, Mar. 1997.

\bibitem{rengaswamy_optimality_2020}
N.~{Rengaswamy}, R.~{Calderbank}, M.~{Newman}, and H.~D. {Pfister}, ``On
  optimality of css codes for transversal {T},'' \emph{IEEE Journal on Selected
  Areas in Information Theory}, vol.~1, no.~2, pp. 499--514, 2020.

\bibitem{knillThresholdAccuracyQuantum1996}
E.~Knill, R.~Laflamme, and W.~Zurek, ``Threshold {{Accuracy}} for {{Quantum
  Computation}},'' \emph{arXiv:quant-ph/9610011}, Oct. 1996.

\bibitem{sarvepalliNonbinaryQuantumReedMuller2005}
P.~Sarvepalli and A.~Klappenecker, ``Nonbinary quantum {{Reed-Muller}} codes,''
  in \emph{Proceedings. {{International Symposium}} on {{Information Theory}},
  2005. {{ISIT}} 2005.}, Sep. 2005, pp. 1023--1027.

\bibitem{bravyi_magic-state_2012}
\BIBentryALTinterwordspacing
S.~Bravyi and J.~Haah, ``Magic-state distillation with low overhead,''
  \emph{Physical Review A}, vol.~86, no.~5, p. 052329, Nov. 2012. [Online].
  Available: \url{https://link.aps.org/doi/10.1103/PhysRevA.86.052329}
\BIBentrySTDinterwordspacing

\bibitem{campbellMagicStateDistillationAll2012}
E.~T. Campbell, H.~Anwar, and D.~E. Browne, ``Magic-{{State Distillation}} in
  {{All Prime Dimensions Using Quantum Reed-Muller Codes}},'' \emph{Physical
  Review X}, vol.~2, no.~4, p. 041021, Dec. 2012.

\bibitem{landahlComplexInstructionSet2013}
A.~J. Landahl and C.~Cesare, ``Complex instruction set computing architecture
  for performing accurate quantum \${{Z}}\$ rotations with less magic,''
  \emph{arXiv:1302.3240 [quant-ph]}, Oct. 2013.

\bibitem{andersonFaultTolerantConversionSteane2014}
J.~T. Anderson, G.~{Duclos-Cianci}, and D.~Poulin, ``Fault-{{Tolerant
  Conversion}} between the {{Steane}} and {{Reed-Muller Quantum Codes}},''
  \emph{Physical Review Letters}, vol. 113, no.~8, p. 080501, Aug. 2014.

\bibitem{kubica_universal_2015}
\BIBentryALTinterwordspacing
A.~Kubica and M.~E. Beverland, ``Universal transversal gates with color codes:
  {A} simplified approach,'' \emph{Phys. Rev. A}, vol.~91, no.~3, p. 032330,
  Mar. 2015. [Online]. Available:
  \url{https://link.aps.org/doi/10.1103/PhysRevA.91.032330}
\BIBentrySTDinterwordspacing

\bibitem{coxeter1973regular}
H.~S.~M. Coxeter, \emph{Regular polytopes}.\hskip 1em plus 0.5em minus
  0.4em\relax Courier Corporation, 1973.

\bibitem{coxeter_regular_1973}
------, \emph{\BIBforeignlanguage{eng}{Regular polytopes}}, 3rd~ed.\hskip 1em
  plus 0.5em minus 0.4em\relax New York: Dover Publications, 1973, open Library
  ID: OL5436086M.

\bibitem{davis_geometry_2007}
M.~W. Davis, \emph{The {Geometry} and {Topology} of {Coxeter} {Groups}.
  ({LMS}-32) ({London} {Mathematical} {Society} {Monographs})}.\hskip 1em plus
  0.5em minus 0.4em\relax Princeton University Press, Oct. 2007, open Library
  ID: OL11182953M.

\bibitem{rotman2008introduction}
J.~J. Rotman, \emph{An introduction to homological algebra}.\hskip 1em plus
  0.5em minus 0.4em\relax Springer Science \& Business Media, 2008.

\bibitem{mcmullen_abstract_2002}
P.~McMullen and E.~Schulte, \emph{Abstract {Regular} {Polytopes}}, ser.
  Encyclopedia of {Mathematics} and its {Applications}.\hskip 1em plus 0.5em
  minus 0.4em\relax Cambridge University Press, 2002.

\bibitem{campbell_unified_2017}
\BIBentryALTinterwordspacing
E.~T. Campbell and M.~Howard, ``Unified framework for magic state distillation
  and multiqubit gate synthesis with reduced resource cost,'' \emph{Physical
  Review A}, vol.~95, no.~2, p. 022316, Feb. 2017. [Online]. Available:
  \url{https://link.aps.org/doi/10.1103/PhysRevA.95.022316}
\BIBentrySTDinterwordspacing

\bibitem{haah_codes_2018}
\BIBentryALTinterwordspacing
J.~Haah and M.~B. Hastings, ``\BIBforeignlanguage{en-GB}{Codes and {Protocols}
  for {Distilling} \${T}\$, controlled-\${S}\$, and {Toffoli} {Gates}},''
  \emph{\BIBforeignlanguage{en-GB}{Quantum}}, vol.~2, p.~71, Jun. 2018.
  [Online]. Available:
  \url{https://quantum-journal.org/papers/q-2018-06-07-71/}
\BIBentrySTDinterwordspacing

\bibitem{Cui_diagonal_2017}
\BIBentryALTinterwordspacing
S.~X. Cui, D.~Gottesman, and A.~Krishna, ``\BIBforeignlanguage{en}{Diagonal
  gates in the {Clifford} hierarchy},'' \emph{\BIBforeignlanguage{en}{Physical
  Review A}}, vol.~95, no.~1, Jan. 2017. [Online]. Available:
  \url{https://link.aps.org/doi/10.1103/PhysRevA.95.012329}
\BIBentrySTDinterwordspacing

\bibitem{jones_multilevel_2013}
\BIBentryALTinterwordspacing
C.~Jones, ``Multilevel distillation of magic states for quantum computing,''
  \emph{Physical Review A}, vol.~87, no.~4, p. 042305, Apr. 2013. [Online].
  Available: \url{https://link.aps.org/doi/10.1103/PhysRevA.87.042305}
\BIBentrySTDinterwordspacing

\bibitem{campbell_unifying_2017}
\BIBentryALTinterwordspacing
E.~T. Campbell and M.~Howard, ``Unifying {Gate} {Synthesis} and {Magic} {State}
  {Distillation},'' \emph{Physical Review Letters}, vol. 118, no.~6, p. 060501,
  Feb. 2017. [Online]. Available:
  \url{https://link.aps.org/doi/10.1103/PhysRevLett.118.060501}
\BIBentrySTDinterwordspacing

\bibitem{hastings_distillation_2018}
\BIBentryALTinterwordspacing
M.~B. Hastings and J.~Haah, ``Distillation with {Sublogarithmic} {Overhead},''
  \emph{Physical Review Letters}, vol. 120, no.~5, p. 050504, Jan. 2018.
  [Online]. Available:
  \url{https://link.aps.org/doi/10.1103/PhysRevLett.120.050504}
\BIBentrySTDinterwordspacing

\bibitem{haah_magic_2017}
\BIBentryALTinterwordspacing
J.~Haah, M.~B. Hastings, D.~Poulin, and D.~Wecker,
  ``\BIBforeignlanguage{en-GB}{Magic state distillation with low space overhead
  and optimal asymptotic input count},''
  \emph{\BIBforeignlanguage{en-GB}{Quantum}}, vol.~1, p.~31, Oct. 2017.
  [Online]. Available:
  \url{https://quantum-journal.org/papers/q-2017-10-03-31/}
\BIBentrySTDinterwordspacing

\bibitem{gallager1962low}
R.~Gallager, ``Low-density parity-check codes,'' \emph{IRE Transactions on
  information theory}, vol.~8, no.~1, pp. 21--28, 1962.

\bibitem{kubica_unfolding_2015}
\BIBentryALTinterwordspacing
A.~Kubica, B.~Yoshida, and F.~Pastawski, ``Unfolding the color code,''
  \emph{New Journal of Physics}, vol.~17, no.~8, p. 083026, Aug. 2015.
  [Online]. Available:
  \url{https://doi.org/10.1088%2F1367-2630%2F17%2F8%2F083026}
\BIBentrySTDinterwordspacing

\bibitem{bhagoji_equivalence_2015}
A.~Bhagoji and P.~Sarvepalli, ``Equivalence of 2d color codes (without
  translational symmetry) to surface codes,'' in \emph{2015 {IEEE}
  {International} {Symposium} on {Information} {Theory} ({ISIT})}, Jun. 2015,
  pp. 1109--1113.

\bibitem{kubica_efficient_2019}
\BIBentryALTinterwordspacing
A.~Kubica and N.~Delfosse, ``Efficient color code decoders in $d\geq 2$
  dimensions from toric code decoders,'' \emph{arXiv:1905.07393 [quant-ph]},
  May 2019, arXiv: 1905.07393. [Online]. Available:
  \url{http://arxiv.org/abs/1905.07393}
\BIBentrySTDinterwordspacing

\bibitem{delfosse_decoding_2014}
\BIBentryALTinterwordspacing
N.~Delfosse, ``Decoding color codes by projection onto surface codes,''
  \emph{Physical Review A}, vol.~89, no.~1, p. 012317, Jan. 2014. [Online].
  Available: \url{https://link.aps.org/doi/10.1103/PhysRevA.89.012317}
\BIBentrySTDinterwordspacing

\bibitem{aloshious_projecting_2018}
\BIBentryALTinterwordspacing
A.~B. Aloshious and P.~K. Sarvepalli, ``Projecting three-dimensional color
  codes onto three-dimensional toric codes,'' \emph{Physical Review A},
  vol.~98, no.~1, p. 012302, Jul. 2018. [Online]. Available:
  \url{https://link.aps.org/doi/10.1103/PhysRevA.98.012302}
\BIBentrySTDinterwordspacing

\bibitem{kubica_three-dimensional_2018}
\BIBentryALTinterwordspacing
A.~Kubica, M.~E. Beverland, F.~Brandão, J.~Preskill, and K.~M. Svore,
  ``Three-{Dimensional} {Color} {Code} {Thresholds} via
  {Statistical}-{Mechanical} {Mapping},'' \emph{Physical Review Letters}, vol.
  120, no.~18, p. 180501, May 2018. [Online]. Available:
  \url{https://link.aps.org/doi/10.1103/PhysRevLett.120.180501}
\BIBentrySTDinterwordspacing

\bibitem{poulin_stabilizer_2005}
\BIBentryALTinterwordspacing
D.~Poulin, ``Stabilizer {Formalism} for {Operator} {Quantum} {Error}
  {Correction},'' \emph{Physical Review Letters}, vol.~95, no.~23, p. 230504,
  Dec. 2005. [Online]. Available:
  \url{https://link.aps.org/doi/10.1103/PhysRevLett.95.230504}
\BIBentrySTDinterwordspacing

\bibitem{bombin_single-shot_2015}
\BIBentryALTinterwordspacing
H.~Bombín, ``Single-{Shot} {Fault}-{Tolerant} {Quantum} {Error}
  {Correction},'' \emph{Physical Review X}, vol.~5, no.~3, p. 031043, Sep.
  2015. [Online]. Available:
  \url{https://link.aps.org/doi/10.1103/PhysRevX.5.031043}
\BIBentrySTDinterwordspacing

\bibitem{delfosse_tradeoffs_2013}
N.~Delfosse, ``Tradeoffs for reliable quantum information storage in surface
  codes and color codes,'' in \emph{2013 {IEEE} {International} {Symposium} on
  {Information} {Theory}}, Jul. 2013, pp. 917--921.

\bibitem{breuckmann_constructions_2016}
N.~P. Breuckmann and B.~M. Terhal, ``Constructions and {Noise} {Threshold} of
  {Hyperbolic} {Surface} {Codes},'' \emph{IEEE Transactions on Information
  Theory}, vol.~62, no.~6, pp. 3731--3744, Jun. 2016.

\bibitem{breuckmann_homological_2017}
\BIBentryALTinterwordspacing
N.~P. Breuckmann, ``Homological quantum codes beyond the toric code,'' {PhD}
  {Thesis}, RWTH Aachen University, 2017. [Online]. Available:
  \url{https://arxiv.org/abs/1802.01520}
\BIBentrySTDinterwordspacing

\bibitem{campbell_theory_2019}
\BIBentryALTinterwordspacing
E.~T. Campbell, ``\BIBforeignlanguage{en}{A theory of single-shot error
  correction for adversarial noise},'' \emph{\BIBforeignlanguage{en}{Quantum
  Science and Technology}}, vol.~4, no.~2, p. 025006, Feb. 2019. [Online].
  Available: \url{https://doi.org/10.1088%2F2058-9565%2Faafc8f}
\BIBentrySTDinterwordspacing

\bibitem{zengHigherDimensionalQuantumHypergraphProduct2019}
W.~Zeng and L.~P. Pryadko, ``Higher-{{Dimensional Quantum Hypergraph-Product
  Codes}} with {{Finite Rates}},'' \emph{Physical Review Letters}, vol. 122,
  no.~23, p. 230501, Jun. 2019.

\bibitem{grassl_codetables_2007}
M.~Grassl, ``{Bounds on the minimum distance of linear codes and quantum
  codes},'' Online available at \url{http://www.codetables.de}, 2007, accessed
  on 2021-04-26.

\bibitem{kesselring_boundaries_2018}
\BIBentryALTinterwordspacing
M.~S. Kesselring, F.~Pastawski, J.~Eisert, and B.~J. Brown,
  ``\BIBforeignlanguage{en-GB}{The boundaries and twist defects of the color
  code and their applications to topological quantum computation},''
  \emph{\BIBforeignlanguage{en-GB}{Quantum}}, vol.~2, p. 101, Oct. 2018.
  [Online]. Available:
  \url{https://quantum-journal.org/papers/q-2018-10-19-101/}
\BIBentrySTDinterwordspacing

\end{thebibliography}

\begin{IEEEbiographynophoto}{Christophe Vuillot} received the Ph.D. degree in computer science from TU Delft, Netherlands, in 2020.
	Since 2021 he is Chargé de Recherche (junior researcher) at Inria, in Nancy, France.
	His research focuses on fault-tolerant quantum computation.
\end{IEEEbiographynophoto}
\begin{IEEEbiographynophoto}{Nikolas P. Breuckmann} obtained the Ph.D. degree at RWTH Aachen University working with Prof. Barbara Terhal on quantum fault-tolerance and quantum complexity theory.
	He has worked in industry at PsiQuantum, a Bay Area-based start-up building a silicon-photonics-based quantum computer. He holds a UCLQ Research Fellowship at University College London.
	He is interested in quantum information and related fields.
\end{IEEEbiographynophoto}

\end{document}